\documentclass[aps,pra,amsmath,amssymb,superscriptaddress,longbibliography,twocolumn]{revtex4-1}
\usepackage{amsmath,bbold}
\usepackage{amssymb}
\usepackage{epsfig}
\usepackage{color}
\usepackage{amsthm}
\usepackage{hyperref}
\usepackage{appendix}
\usepackage{cancel}
\usepackage{stmaryrd}
\usepackage{soul}

\def\<{\langle}
\def\>{\rangle}

\newcommand{\be}{\begin{eqnarray} \begin{aligned}}
\newcommand{\ee}{\end{aligned} \end{eqnarray} }
\newcommand{\benn}{\begin{eqnarray*} \begin{aligned}}
\newcommand{\eenn}{\end{aligned} \end{eqnarray*} }

\newcommand{\ben}{\begin{eqnarray} \begin{aligned}}
\newcommand{\een}{\end{aligned} \end{eqnarray} }

\newcommand{\bc}{\begin{center}}
\newcommand{\ec}{\end{center}}

\newcommand{\id}{\mathbb{I}}

\newcommand{\tr}{\mathop{\mathsf{tr}}\nolimits}

\newcommand{\e}{\mathrm{e}}

\newcommand{\beq}{\begin{eqnarray} \begin{aligned}}
\newcommand{\eeq}{\end{aligned} \end{eqnarray} }
\newcommand{\bea}{\begin{array}}
\newcommand{\eea}{\end{array}}

\newcommand{\bee}{\begin{enumerate}}
\newcommand{\eee}{\end{enumerate}}
\newcommand{\bei}{\begin{itemize}}
\newcommand{\eei}{\end{itemize}}
\newtheorem{theorem}{Theorem}

\newtheorem{lemma}[theorem]{Lemma}

\newtheorem{result}{Result}

\usepackage{amsfonts}

\def\id{\mathbb{I}}

\def\01{\{0,1\}}

\newcommand{\ket}[1]{|#1\rangle}
\newcommand{\bra}[1]{\langle#1|}

\newcommand{\proj}[1]{|#1\rangle\!\langle#1|}
\newcommand{\ketbra}[2]{|#1\rangle\!\langle#2|}

\newcommand{\braket}[2]{\langle #1|#2\rangle}

\newcommand{\supl}{Appendix}

\def\<{\langle}
\def\>{\rangle}

\newtheorem*{rep@theorem}{\rep@title}
\newcommand{\newreptheorem}[2]{%
\newenvironment{rep#1}[1]{%
 \def\rep@title{#2 \ref{##1} (restatement)}%
 \begin{rep@theorem}}%
 {\end{rep@theorem}}}
\makeatother

\newreptheorem{thm}{Theorem}
\newreptheorem{lem}{Lemma}

\def\e{\mathrm{e}}

\def\bat{\eta}

\newcommand{\alv}[1]{{ #1}}

\begin{document}

\title{Entanglement fluctuation theorems}

\author{\'{A}lvaro M. Alhambra}
\affiliation{Department of Physics and Astronomy, University College London, Gower Street, London WC1E 6BT, United Kingdom}
\affiliation{Perimeter Institute for Theoretical Physics, Waterloo, ON N2L 2Y5, Canada}
\author{Lluis Masanes}
\affiliation{Department of Physics and Astronomy, University College London, Gower Street, London WC1E 6BT, United Kingdom}
\author{Jonathan Oppenheim}
\affiliation{Department of Physics and Astronomy, University College London, Gower Street, London WC1E 6BT, United Kingdom}
\author{Christopher Perry}
\affiliation{QMATH, Department of Mathematical Sciences, University of Copenhagen, Universitetsparken 5, 2100 Copenhagen, Denmark}

\begin{abstract}
 Pure state entanglement transformations have been thought of as irreversible, with reversible transformations generally only possible in the limit of many copies.
 Here, we show that reversible entanglement transformations do not require processing on the many copy level, but can instead be undertaken on individual systems, provided the amount of entanglement which is produced or consumed is allowed to fluctuate.
 We derive necessary and sufficient conditions for entanglement manipulations in this case. As a corollary, we derive an equation which quantifies the fluctuations of entanglement, which is formally identical to the Jarzynski fluctuation equality found in thermodynamics. One can also relate a forward entanglement transformation to its reverse process in terms of the entanglement cost of such a transformation, in a manner equivalent to the Crooks relation. We show that a strong converse theorem for entanglement transformations is formally related to the second law of thermodynamics, while the fact that the Schmidt rank of an entangled state cannot increase is related to the third law of thermodynamics. 
 Achievability of the protocols is done by introducing an {\it entanglement battery}, a device which stores entanglement and uses an amount of entanglement that is allowed to fluctuate but with an average cost which is still optimal. This allows us to also solve the problem of partial entanglement recovery, and in fact, we show that entanglement is fully recovered. Allowing the amount of consumed entanglement to fluctuate also leads to improved and optimal entanglement dilution protocols.
\end{abstract}
\maketitle

\section{Introduction}

Entanglement is generally regarded as the essential feature of quantum mechanics. Originally introduced by Einstein, Podolsky and Rosen~\cite{epr} to argue that quantum mechanics was not a complete theory of nature, and shown by Bell to not be explainable by any {\it locally realistic} theory~\cite{bell}, it is now regarded as the key resource in quantum information theory. It allows for basic primitives such as teleportation and quantum cryptography, is seen as a key ingredient in the speed-up of quantum computers, and is behind quantum advantages in communication complexity and precision measurements. The pioneering works of quantum information theory sought to quantify entanglement~\cite{BBPS1996,BDSW1996,BBPSSW1996}, and provide conditions for transformations between entangled states using only Local Operations and Classical Communication (LOCC) in the limit of sharing many copies of the same state. In these works, one can think of entanglement as an average quantity, with entanglement manipulations generally only being possible in the limit of sharing many copies of the same quantum state.

This mirrors the early stages of the history of statistical mechanics a century ago, when quantities such as work and heat, while regarded as being a single number,
are really average quantities which only emerge in the thermodynamic limit. Indeed, the analogy between thermodynamics and pure state entanglement transformations was made explicit in \cite{popescu-rohrlich}, as well as in the case of mixed state entanglement manipulations with more limited success~\cite{termo,APE,horodecki_are_2002,brandao2008entanglement,li2017asymptotic}.

At around the same time as the resource theory of entanglement was being developed, the Jarzynski equation, and Crooks relation were discovered. These and other results in Stochastic Thermodynamics \alv{give exact information} about the fluctuations of work and heat about their average values \cite{jarzynski1997nonequilibrium,crooks1999entropy,sagawa2010generalized,seifert2012stochastic}. This raises the question as to whether one can understand entanglement as a fluctuating quantity, from which we can understand some of the present many copy results in entanglement theory as a restriction to the case when we can only compute average quantities. Indeed, it has recently been shown that there is a connection between fluctuation theorems and the majorization condition \cite{alhambra2016fluctuating}, raising the prospect that fluctuation theorems have wider applicability. This is because the majorization criteria and its generalization is known to play an important role in both determining state transformations in thermodynamics \cite{horodecki_reversible_2003,horodecki2013fundamental}, and in single-copy pure state entanglement manipulation \cite{Nielsen-pure-entanglement}.

Here, we will see that we can in fact think of entanglement as a resource whose amount can fluctuate and we derive a fluctuation theorem which quantifies the extent to which it can. In order to do this, we will need to define what we mean by fluctuations of entanglement. In some cases, we may be interested in processes which, with probability $P\left(w\right)$, coherently produce or consume some amount $w$ of maximally entangled pure states (or {\it ebits}). We will find necessary and sufficient conditions that this superposition of entanglement fluctuations has to satisfy. Our conditions apply to all pure state transformations, even those which probabilistically create a pure target state from some ensemble. To achieve the conditions, we introduce the notion of an entanglement battery, which is a system which stores entanglement and introduce and prove the existence of the family of \emph{battery-assisted} LOCC protocols. The operations allowed on such battery can add or consume entanglement from it in a coherent manner, and the necessary and sufficient conditions we derive will characterize this superposition.

In doing so, we find that this entanglement battery can be used to perform tasks which were previously impossible. For example, pure state entanglement transformations, which are generally irreversible at the level of single copies \cite{Nielsen-pure-entanglement}, become reversible. A special case of this reversibility is entanglement {\it concentration} and {\it dilution}, two of the most basic primitives of entanglement theory. In concentration, many copies of a pure state are converted into many maximally entangled states while dilution is the reverse process. Current protocols only work in the asymptotic limit of infinitely many copies, and while current concentration protocols are optimal, dilution protocols are not. For general pure state transformations, the rate of converting $n$ copies of state $\ket{\psi}_{AB}$ into $m$ copies of the state $\ket{\phi}_{AB}$ is given by
\begin{align}
\frac{n}{m}=\frac{S\left(\psi_A\right)}{S\left(\phi_A\right)}
\label{eq:conversionrate}
\end{align}
with  $S\left(\rho\right):=-\tr\rho\log\rho$, $\psi_A=\tr_B\proj{\psi}_{AB}$ and similarly for $\phi_A=\tr_B\proj{\phi}_{AB}$. In this asymptotic limit, state transitions become reversible, but only up to factors of order $\sqrt{n}$.  So, while $\ket{\psi}_{AB}^{\otimes n}\rightarrow \ket{\phi}_{AB}^{\otimes m}$ may be possible by LOCC, it is generally the case that $\ket{\phi}_{AB}^{\otimes m}\rightarrow\ket{\psi}_{AB}^{\otimes n-o(\sqrt{n})}$. Or, to put it another way, in the limit of large $n$, the transition $\ket{\psi}_{AB}^{\otimes n}\rightarrow \ket{\phi}_{AB}^{\otimes n}$ is possible, consuming (producing)  $S(\phi_A)-S(\psi_A)$ maximally entangled states (ebits) on average, while the reverse process $\ket{\phi}_{AB}^{\otimes n}\rightarrow\ket{\psi}_{AB}^{\otimes n}$ is possible with the production (consumption) of the same average number of ebits up to factors of $\frac{\sqrt{n}}{n}$. In the limit of large $n$, the rates for the forward and reverse process are the same, but the difference between the absolute number of ebits used diverges.

Although the average rate of entanglement consumed or produced is given by $S(\phi_A)-S(\psi_A)$, this quantity will fluctuate and in any instance of the protocol, the amount of maximally entangled states which one obtains or consumes will vary. Here, we derive \alv{a number of} fluctuation theorems which exactly characterize these fluctuations. Thus far, the characterization of such fluctuations has been unsolved, with the only progress being that one can obtain the probability of concentrating to $m$ ebits in the regime of infinitely many input states, provided no constraints are put on the rest of the distribution \cite{LoPopsecu1997-beyond}.

Indeed, the result in \cite{LoPopsecu1997-beyond} can be seen as the majorization condition in the special case of the final state being maximally entangled. More generally, Nielsen \cite{Nielsen-pure-entanglement} showed that 
\begin{align}
\ket{\psi}_{AB}\rightarrow \ket{\phi}_{AB}
\label{eq:puretrans}
\end{align}
is possible by LOCC, if and only if, the majorization condition $q\left(\phi\right)\succ p\left(\psi\right)$ holds, i.e. that
\begin{align}
\sum_{j=1}^{k} q_j(\phi) \geq \sum_{i=1}^{k} p_i(\psi),\quad \forall k
\end{align}
with $p_i$, $q_j$ being the eigenvalues of $\psi_A$, $\phi_A$ written in non-increasing order $p_1\geq p_2\geq \dots \geq p_d$. Indeed, the LOCC protocol which achieves any pure state transformation can be taken to consist of a POVM measurement by Alice, followed by a unitary transformation by Bob conditional on the result of Alice's POVM \cite{LoPopsecu1997-beyond}. That majorization is a necessary condition for a pure state transformation suggests that single copy transformations are irreversible, and typically, $q(\phi)\nsucc p(\psi)$ and  $q(\phi)\nprec p(\psi)$ \cite{Nielsen-pure-entanglement}, meaning that no transition can happen in either direction.

Surprisingly, we find this is not the case in the presence of an entanglement battery. We will see that one can perform any pure state entanglement transformation at the single copy level. Furthermore, we see that reversibility on the single-copy level is restored in the presence of an entanglement battery, and that we can exactly characterize the fluctuations of entanglement in the battery. As a result, one does not need an infinite number of copies of the input state to distill entanglement, but instead the copies can be processed on the individual level. This special case is reminiscent of the streaming entanglement distillation protocols introduced in \cite{blume2009streaming}, where processing is done one system at a time, albeit with a quantum memory of order $\log{n}$. In the streaming protocol, ebits are emitted after a lag of $\log{n}$ states have been processed, thus the lag becomes infinite in the limit where we achieve perfect entanglement concentration. Using the entanglement battery, there is no lag, and the processing is truly on the individual copy level.

\alv{The paper is structured as follows. In Section \ref{sec:batteryBLOCC} we define the notion of an entanglement battery and of the operations on it that yield entanglement fluctuations. Then, in Section \ref{sec:results} we state all the main results, which take the form of necessary and sufficient conditions on those fluctuations. We first start from the more general set of them, to them move to more specific constraints akin to the so-called \emph{integral} fluctuation theorems. We finish the section with the analogue of Crooks' theorem for which a reverse process is defined. We finally conclude in Section \ref{sec:conclusion}, where we explain applications of our results to open problems in entanglement theory such as \emph{partial entanglement recovery} or \emph{embezzlement}, and we briefly discuss the experimental feasibility of the protocols. We place the proofs of the main results as well as further details about the setting in the Appendix.}

\section{Entanglement battery and battery-assisted LOCC}\label{sec:batteryBLOCC}

Let us now introduce the notion of an entanglement battery, and define the ways in which one can act on it. \alv{In analogy with the thermodynamic scenario \cite{alhambra2016fluctuating}, the definition of entanglement ``work" will be determined by the restriction we impose on the type of transformations we implement on the battery.}

 Just as an ordinary battery (such as a weight at height $h$) stores energy which can be used to inject or store work in the context of thermodynamics, the entanglement battery can be thought of as a storage device for entanglement. Just as work can be thought of as the change in average energy of the battery or average height of the weight, we will see that von Neumann entropy can be thought of as the change in the average number of maximally entangled states stored in the entanglement battery.

\alv{We are } interested in characterizing the entanglement fluctuations of any state transformation. To do so, let us begin by considering a natural set of battery states $\ket{e_x}$ on a system $A'B'$ given by
\begin{equation}
\label{integer battery}
\ket{e_x}_{A'B'}=\left(\frac{1}{\sqrt{2}}\left(\ket{1}\ket{1}+\ket{2}\ket{2}\right)\right)^{\otimes x}\otimes \left(\ket{0}\ket{0}\right)^{\otimes\left(n-x\right)},
\end{equation}
for some large $n$, i.e. they consist of $x$ ebits, and $n-x$ pure product states. A common subset of LOCC protocols that we want to characterise are those which coherently produce or consume ebits with some probability.

Such protocols can be considered by taking $A'B'$ to be a battery which starts off in state $\ket{e_x}_{A'B'}$, and then, if $w$ ebits of entanglement are added or removed from the battery, $\ket{e_x}_{A'B'}\rightarrow\ket{e_{x+w}}_{A'B'}$ provided $x$ and $n$ are sufficiently large so that we avoid the top and bottom of the battery. We will use the convention that positive $w$ corresponds to gaining entanglement, while negative $w$ corresponds to consuming it. Just as work is the raising and lowering of the weight, here, we want to consider the raising and lowering of the number of ebits by $w$ and we are interested in characterising the fluctuations in $w$ during the pure state LOCC transformation of Eq.~\eqref{eq:puretrans}.

We thus want to consider entanglement gain/consumption, to be the process of raising/lowering the number of ebits in the entanglement battery with the raising/lowering operator defined through $\Delta^w \ket{e_x}_{A'B'}=\ket{e_{x+w}}_{A'B'}$, where $x+w$ is understood modulo $n+1$ to ensure that $\Delta^w$ is a unitary (though we will pick the states on the battery to be such that the top and bottom of the battery are never reached in practice). If initially, the battery is found to be in state $\ket{e_x}_{A'B'}$, then the final superposition over $\ket{e_{x+w}}_{A'B'}$, $\sum_w \sqrt{P\left(w\right)}\ket{e_{x+w}}_{A'B'}$, gives us a probability distribution over entanglement we call $P(w)$. Note also that either Alice or Bob can measure the amount of entanglement in the battery resulting in the state $\ket{e_{x+w}}$ with probability $P\left(w\right)$ and revealing the entanglement loss or gain. In general, we might want $w$ to take on non-integer values, and indeed one can easily consider a set of battery states which allow this, as discussed in Appendix \ref{ss:battery}.


In our protocols, we will consider more general initial battery states of the form
\begin{align}
\ket{\bat}_{A'B'}=\sum_{x=0}^{n} \sqrt{\alpha_x} \ket{e_x}_{A'B'}
\ .
\label{eq:battery}
\end{align}
What we require from the state of our battery $\ket{\bat}_{A'B'}$ is that, for any pure input $\ket\psi$ and output $\ket\phi$ states of the system, the LOCC transformation 
\begin{align}
\ket{\psi}_{AB}\otimes\ket{\bat}_{A'B'} \longrightarrow \sigma_{ABA'B'}\approx \ket{\phi}_{AB}\otimes \ket{\bat'}_{A'B'}
\label{eq:puretransbat}
\end{align}
can be achieved reversibly, with $\ket{\bat'}_{A'B'}$ being a state of the battery which is also useful for further arbitrary entanglement transformations. This is a fairly strong condition, because in order to ensure that the final state of the system is pure, it must be virtually uncorrelated with the battery.

\alv{In fact, we show} that purity of the target state of the system implies that in the limit of ideal transformations, the battery is in a uniform superposition over $\ket{e_x}_{A'B'}$ (see subsection \ref{ss:uniform} of the \supl \,for the proof). 
This is \alv{qualitatively} similar to the case of using a reference frame in order to perform a transformation on pure states which   would otherwise be limited by a conservation law \cite{aharonov-susskind,BRS-refframe-review}. It might appear surprising that one can retain purity on the system, since the battery would appear to become correlated with it. However, as we show in Appendix \ref{ssec:large n} not only can this be done, but also it can be done perfectly, as
\begin{equation}
  \label{eq:purenochange}
  \ket{\psi}_{AB}\otimes\ket{\bat}_{A'B'} \rightarrow \sigma_{ABA'B'}\approx \ket{\phi}_{AB}\otimes \ket{\bat}_{A'B'}
\end{equation}
provided the battery state is chosen to be close to a uniform superposition over sufficiently many eigenstates $\ket{e_x}$, for example, $\alpha_x=\frac{1}{N+1}$ for $x\in\left\{\frac{n-N}{2},\dots,\frac{n+N}{2}\right\}$ with $N$ chosen large enough to obtain the transformation of Eq.~\eqref{eq:purenochange} to any desired accuracy. For general state transformations, we must therefore take the initial state of the battery to be in such a uniform superposition, and the final state of the battery must also be close to this if it is to be used for further arbitrary state transformations.

The raising/lowering maps set the type of transitions that the battery can undergo, and achieving Eq. \eqref{eq:purenochange} for all states is a consequence of this definition. These maps are guaranteed to exist but we note, however, that we have not found whether they can have a fixed form as a LOCC map on the battery that applies the transition of Eq. \eqref{eq:purenochange} universally. This would amount to finding an explicit form of the LOCC protocols that raise and lower the battery with the operator $\Delta$ and that implements any one of the possible state transitions on the system.

That a transformation of the form of Eq.~\eqref{eq:purenochange} is possible might appear paradoxical, since the entanglement in the battery is changing, but the state of the battery barely does. However, large changes in average quantities need not correspond to large changes in the state. In particular, $\bra{\bat} \Delta^w\ket{\bat}_{A'B'}$ is close to $1$ for all $w$, and thus the states of the system will not become correlated with the battery. Nonetheless, the average entanglement of the states $\Delta^w\ket{\bat}_{A'B'}$ and $\ket{\bat}_{A'B'}$ differ by $w$, reflecting the fact that large changes in a system's average observables need not take it to orthogonal states. This is also a property of embezzling states \cite{Hayden-embezzling,harrow_entanglement_2009,leung2008coherent}, although the processes we consider are more general than embezzling transformations, as they generally require classical communication to perform, while embezzlement does not~\cite{Hayden-embezzling}.

A similar phenomenon also occurs in the use of reference frames to maintain coherence. For instance, in \cite{aberg2014catalytic} it was shown that a large uniform superposition over energy levels can be used repeatedly to simulate arbitrary unitary processes on a single small system via energy-preserving interactions. The results here are qualitatively similar, in that we use a large entanglement battery in a uniform superposition to perform transitions that would otherwise be restricted. However, the constraints that must be circumvented in both scenarios are formally very different: here we have to circumvent the majorization constraints imposed by Nielsen's theorem \cite{Nielsen-pure-entanglement}, while in \cite{aberg2014catalytic} one is limited by asymmetry considerations, akin to the WAY theorem \cite{marvian2012information,ahmadi2013wigner} (in which majorization does not appear). 

We will call any LOCC protocol which implements Eq.~\eqref{eq:purenochange}  {\it Battery assisted LOCC}, or {\it BLOCC}, which will be the set of operations allowed in what follows. More precisely, we say that a pure state transformation is possible under BLOCC if there exists a sequence of BLOCC protocols $\Lambda_N$ and states $\ket{\Psi_N}$ and $\ket{\Phi_N}$ such that $\ket{\Psi_N} \rightarrow\ket{\Phi_N}$ under BLOCC and $\lim_{N\rightarrow\infty}\ket{\Psi_N}=\ket{\psi}_{AB}\otimes\ket{\bat}_{A'B'}$ and $\lim_{N\rightarrow\infty}\ket{\Phi_N}=\ket{\phi}_{AB}\otimes\ket{\bat}_{A'B'}$. In what follows, we will state our results in the limit of large $N$, although in our proofs we consider the finite case. 

Although purity of the target state requires a perfectly uniform battery, in Appendix \ref{ssec:large n} we show how to relax this condition \alv{to batteries of finite size. More specifically}, we show how to implement the more general state transformation:
\begin{align} \label{eq:goal}
&\ket{\Psi}_{ABA'B'}:=\sum_{i,x} \sqrt{p_{ix}}\ket{ii}\otimes \ket{e_x} \\ \nonumber &\stackrel{\Lambda}{\longrightarrow} \ket{\Phi}_{ABA'B'}:= \sum_{j,x'} \sqrt{q_{jx'}}\ket{jj}\otimes \ket{e_{x'}},
\end{align}
via a BLOCC protocol $\Lambda$ and with the initial and final state of the system being arbitrarily close to pure. As before, the probability distribution over
entanglement fluctuations can still be quantified by imagining that initially, we could have measured $\ket{e_x}$, and at the end of the process we could have measured  $\ket{e_{x'}}$, and we consider the entanglement fluctuation in the battery to be given by $w=x'-x$ with probability $P\left(w\right)$ \alv{given by the statistics of those measurements}.

Indeed, we can sample from the probability distributions $p_i$, $q_j$, and $P\left(w\right)$, as well as the joint distribution $P(i,j,w)$ as follows: Initially, Bob (or Alice) can measure his (her) state with the POVM $\left\{\ketbra{i}{i}\otimes P_x\right\}$ where $P_x$ is the projector onto the subspace spanned by the reduced state of $\ket{e_x}$. Alice then performs the POVM measurement which would have implemented the transformation of Eq.~\eqref{eq:goal}, and, finally, measures her state with the POVM $\left\{\ketbra{j}{j}\otimes P_{x'}\right\}$. Note that Alice's POVM commutes with Bob's measurement, and Alice's measurement of $x'$ commutes with Bob's measurement of $x$, and so we can compute $P(i,j,w)$. When these initial and final  measurements are performed the actual transformation $\ket{\psi}_{AB}\rightarrow \ket{\phi}_{AB}$ does not happen, but their statistics capture the relevant information of the map $\Lambda$.

\alv{Maybe here we can discuss what the ref 3 suggests}

\section{Results}\label{sec:results}

\alv{We now present the main results, which take the form of constraints on the possible transitions via BLOOC, given by relations between the fluctuations in the battery (labeled by $w$) and the Schmidt coefficients of the initial and final states of the system.}
We are able to prove six results about fluctuations of entanglement and use them to prove the existence of optimal BLOCC protocols for entanglement dilution and partial entanglement recovery. The main result from which the remaining five follow gives a family of necessary and sufficient conditions for state transformations to be possible under BLOCC, which take a similar form to those in \cite{alhambra2016fluctuating} in the context of quantum thermodynamics (a  major difference being the role of initial and final states in the constraints, which is reversed):
\begin{result}[Assisted stochasticity]\label{result:bat-stochastic}
A pure state BLOCC transformation $\ket{\psi}\rightarrow\ket{\phi}$ between states with Schmidt coefficients $p_i$ and $q_j$ and a distribution of maximal entanglement $P\left(w\right)$ coherently consumed or produced in the process, is possible if and only if there exists a conditional probability distribution $P\left(i,w|j\right)$ satisfying the following three conditions:
\begin{equation}
  \sum_{i,w} P\left(i,w|j\right)=1,\quad \forall j.
  \label{eq:normalisation}
\end{equation}
\begin{equation}
  \sum_{j,w} P\left(i,w|j\right)2^w=1,\quad \forall i.
  \label{eq:w-stochasticity}
\end{equation}
\begin{equation}
  \sum_{j,w} P\left(i,w|j\right)q_j=p_i, \quad \forall i.
  \label{eq:doesthetrick}
\end{equation}
\end{result}
We give the proof of the necessity and sufficiency of these conditions in Appendix~\ref{sec:necc} and \ref{sec:suff} respectively. The conditions Eqs.~\eqref{eq:normalisation}-\eqref{eq:doesthetrick} can be thought of as a generalization of the doubly-stochastic conditions for the matrix of $\Lambda$, which is well known to be equivalent to the standard majorization condition on the initial and final state. The appearance of the random variable $w$ reflects the departure from this double-stochasticity (recovered when $w=0$), and thus the non-zero values of $w$ reflect our ability to perform transitions on the system beyond those allowed by the usual majorization constraints.
This again is due to the use of the battery through the BLOCC protocols (which are a subset of all the possible LOCC protocols on system and battery, with the properties outlined in Section \ref{sec:batteryBLOCC}).

Our next result can be derived from the above relations (see Appendix~\ref{sec:2nd eq} for proof) and can be though of as the second law of entanglement:
\begin{result}[The 2nd law equality for entanglement] \label{res:2nd eq}
Given an initial state $\ket{\psi}$ with Schmidt coefficients $p_i$ and a target state $\ket{\phi}$ with coefficients $q_j$, the distribution of entanglement that can be coherently extracted in converting $\ket{\psi}$ into $\ket{\phi}$ under BLOCC satisfies:
\begin{equation} \label{eq:entSLE}
\left\langle 2^{w-\log q_j+\log p_i}\right\rangle=1.
\end{equation}
\end{result}

This equality is akin to recent fluctuation theorems for arbitrary input and output states \cite{funo2013integral,schumacher2011quantum,sagawa2012second,seifert2012stochastic,alhambra2016fluctuating}. 

The next result is a single necessary and sufficient condition for a transformation between states via BLOCC protocols, provided one has access to enough fluctuating entanglement.

\begin{result}[Conditions for state transformations with entanglement fluctuations]\label{res:rev}
 The transformation $\overset{\langle w \rangle}{\ket{\psi}\rightarrow\ket{\phi}}$ is possible under BLOCC, if and only if
\begin{equation} \label{eq:inequality}
	\langle w \rangle \leq S(\psi_A)-S(\phi_A).
\end{equation}
\end{result}
\alv{To prove the necessity, we just need to use Jensen's theorem on Eq.~\eqref{eq:entSLE}, together with the fact that the exponential function  is convex.
Sufficiency follows from setting a particular value to the work fluctuations, $w=\log q_j-\log p_i$ in Eq.~\eqref{eq:w-stochasticity}, which gives a set of work fluctuations that saturates the inequality}. This inequality can be thought of as akin to the traditional second law of thermodynamics, stated that the average work $W$ required in transforming a state $\rho$ into a state $\sigma$ has to satisfy $\left\langle W\right\rangle \leq F(\rho)-F(\sigma)$ with $F(\rho)$ the free energy $F\left(\rho\right)=\langle H\rangle-TS(\rho)$, $T$ the temperature of the bath that the system is in contact with and $H$ the Hamiltonian of the system.

It is this result which implies reversibility of single copy transformations if an entanglement battery is allowed. This is the same sense in which thermodynamics has a reversible regime. Going between two states in thermodynamics requires an amount of work given by the change in free energy, while in the reverse process one obtains back the same amount of work.  Here, we even have reversibility on the level of fluctuations, namely, if we have entanglement fluctuations $w=\log{q_i}-\log{p_i}$, then there exists a reverse process which has equal and opposite entanglement fluctuations given by $w_{rev}=\log{p_i}-\log{q_i}$.

This contrast -- between irreversibility of single copy transformations when one does not allow a battery, and the ability to perform such transformations when one has an entanglement battery -- is reminiscent of very recent results in thermodynamics. There, it has been shown that state transformations which occur at the small scale \cite{HO-limitations} are fundamentally irreversible. Yet, if one allows fluctuating work, reversible transformations are possible when acting on single copies \cite{skrzypczyk2014work}. 

In addition to Eq.~\eqref{eq:inequality}, higher order corrections governing entanglement manipulations can be found by Taylor expanding Eq.~\eqref{eq:entSLE} as in \cite{alhambra2016fluctuating} giving:
\begin{equation}\label{eq:infbounds}
\sum_{k=1}^M
\frac {\left(\ln2\right)^k} {k!} 
\left\langle \left(w-\log q_j +\log p_i \right)^k \right\rangle 
\leq 0\ ,
\end{equation}
with $M$ odd and $M=1$ corresponds to the previously known average case.

\alv{That Eq. \eqref{eq:inequality} is necessary and sufficient means any state transition is in fact possible, given that one has access to enough entanglement \emph{on average}. However, this constraint does not contain information about the size such fluctuations, which also has to be taken into account. 
To show their importance, we focus on pure to pure transitions in which the Schmidt rank increases. Without a battery, the Schmidt rank cannot increase, not even probabilistically \cite{vidal1999entanglement}, and when assisted with a battery in BLOCC protocols, the difficulty of such transitions is expressed in the following result.
\begin{result}[Third law of entanglement]\label{re:3rdlaw}
	Let $p_{\text{min}}$ and $q_{\text{min}}$ be the smallest Schmidt coefficients of the initial and final states of the system. The entanglement fluctuations are bounded by
	\begin{equation}\label{eq:3rdlawin}
	\sum_w 2^w \ge \frac{q_{\text{min}}}{d' p_{\text{min}}},
	\end{equation}
	where $d'$ is the number of nonzero Schmidt coefficients of the final state.
\end{result}
The proof of this result follows from considering a particular case of the constraints of Result \ref{result:bat-stochastic}. The details can be found in Section \ref{sec:third} of the Appendix.
From this, it follows that in the limit in which we are increasing the Schmidt rank (that is, when $p_{\text{min}}\rightarrow 0$), either the amount of fluctuations or the size of them must diverge.
This is the analogue of a number of results in thermodynamics associated with the 3rd law, which roughly speaking states that decreasing the rank of a state requires infinite resources, in the form of infinite work fluctuations, an infinite-sized bath or both \cite{masanes2017general}.

On the other hand, other transitions between states may be such that the initial one has Schmidt coefficients that majorize the final. In that case one can refer back to the setting of Nielsen's theorem \cite{Nielsen-pure-entanglement}, which shows that no work at all is needed for the transition. In such cases in which the majorization conditions hold, perhaps up to some small error, we expect that the size of the fluctuations of $w$ will not be very strongly constrained. For instance, one could have a transition allowed by the majorization criteria in which some entanglement is gained on average, or where some large fluctuations occurs with a small probability.

 This discussion indicates that even if in principle all transitions between states are possible given enough average entanglement (as per Result \ref{res:rev}), some may still be harder than others if one considers the size or number of those fluctuations. This information is not captured by Eq. \eqref{eq:inequality} but by the necessary and sufficient constraints on the stochastic matrices of Result \ref{result:bat-stochastic}. 
 
 One can think of the appearance of $w$ in Eqs.~\eqref{eq:normalisation}-\eqref{eq:doesthetrick} as a correction to the bistochasticity constraints imposed by majorization. Thus, we expect that the further a particular transition is from satisfying the majorization conditions, the larger the work fluctuations should be to allow for that transition. Both Eq. \eqref{eq:inequality} and \eqref{eq:3rdlawin} support this conclusion in a different way: Eq. \eqref{eq:inequality} says that if the entanglement entropy of the initial state is lower than that of the final (forbidden by majorization), average fluctuations are unavoidable, and Eq. \eqref{eq:3rdlawin} says that if $q_{\min}/p_{\min}$ is very large (also forbidden by majorization), either we have a large number of fluctuations or a small number of very large ones.}

The general necessary and sufficient constraints of Result \ref{result:bat-stochastic} allow us to also find an analogue of the Jarzynski equation which applies to the case where the final state $\ket{\phi}_{AB}$ is a maximally entangled state of dimension $d'$
\begin{result}[Jarzynski for entanglement] \label{res:jar}
When the final state is a maximally entangled states of dimension $d'$, we have:
\begin{equation} \label{eq:entjar}
\left\langle 2^w \right\rangle = \frac{d}{d'},
\end{equation}
with $d$ the dimension of the support of the initial state.
\end{result}
The proof is given in Appendix~\ref{sec:jar} and it follows easily from the constraints of Result \ref{result:bat-stochastic}. Recall that the Jarzynski equation applies when an initially thermal state is driven to an out of equilibrium state, with a possibly different Hamiltonian. It is written as  
\begin{equation}
  \label{eq:thermo-jar}
  \left\langle
  \e^{\beta W}
  \right\rangle 
  = \frac {Z'} {Z}\ ,
\end{equation}
where $W$ is the thermodynamic work, and $Z$ and $Z'$ are the initial and final partition functions $Z=\tr[e^{-\beta H}]$. We thus see that for entanglement, the dimension of a maximally entangled state is akin to the partition function of the thermal state.

An immediate application of Eq.~\eqref{eq:entjar} is that it provides a {\it strong converse bound} \alv{that applies when} one attempts to concentrate more entanglement than the maximum rate given by $\log{d/d'}$ \cite{LoPopsecu1997-beyond}. That is, if one attempts to extract more entanglement than that rate, one immediately sees that in order to satisfy Eq.~\eqref{eq:entjar}, the probability of success has to go exponentially quickly to zero. 
\begin{align}
  P\!\left( w \ge \log{\mbox{$\frac{d}{d'}$}}+x \right)
  &=
  \sum_{w \ge \log{\frac{d}{d'}}+x}
  \hspace{-4mm} P(w)\  
 \\& \le 
  \sum_{w \ge \log{\frac{d}{d'}}+x}
  \hspace{-4mm} P(w)\ 2^{w- \log{\frac{d}{d'}}-x}\ \nonumber
 \\& \le \ 
  \sum_{w} P(w)\ 
  2^{w- \log{\frac{d}{d'}}-x}
  = 2^{-x}.\nonumber
\end{align}
 In the thermodynamic case, one similarly has that Eq~\eqref{eq:thermo-jar} implies that if one attempts to extract work from a heat bath, the probability of success goes exponentially fast to zero, and it is thus a quantitative strengthening of the ordinary second law of thermodynamics, which simply says that the average work you can extract from a single heat bath in a cyclic process is zero. \alv{With this,} we point out a link between the second law and the strong converse.

This brings us to our fifth result, an analogue of the Crooks relation from statistical mechanics~\cite{crooks1999entropy}, which we explain in Section \ref{sec:appcrooks} for completeness. Given any \emph{forward} LOCC protocol corresponding to the matrix $P\left(i,w|j\right)$,  we are able to define a reverse LOCC protocol where, in particular, Bob's unitary transformations are taken to be the inverse of the forward ones. \alv{The reverse process and its relation with the forward one are explained in detail in Section \ref{sec:reversed} of the Appendix. As it turns out, the two processes are related in a way similar to how a process and its time-reversed analogue are related in thermodynamics} \cite{crooks2008quantum,aberg2016fully}. In fact, we find:
\begin{result}[Crooks for entanglement]
Suppose $\ket{\psi}\stackrel{\textrm{BLOCC}}{\longrightarrow}\ket{\textrm{ebit}_d}$ while extracting entanglement $\left\{P\left(w\right),w\right\}$. Then there exists a state $\ket{\psi'}$ such that $\ket{\psi'}\stackrel{\textrm{BLOCC}^{\textrm{rev}}}{\longrightarrow}\ket{\textrm{ebit}_{d'}}$ while extracting entanglement $\left\{P^{\textrm{rev}}\left(-w\right),w\right\}$ and where:
\begin{equation}
\frac{P\left(w\right)}{P^{\text{rev}}\left(-w\right)}=2^{-w} \frac{d'}{d}.
\end{equation}
\end{result}
\alv{
The proof of this statement follows straightforwardly once the definition of the reverse process is established. It can be found in Section \ref{sec:crooks2} of the Appendix.}

In the same way that the Jarzynski equality can be derived from Crooks' theorem, this expression is a refinement of Eq.~\eqref{eq:entjar}. It implies that extracting $w$ in a forward protocol is exponentially suppressed with respect to extracting $-w$ in the reversed protocol. 


\section{Conclusion}\label{sec:conclusion}

Our results show that, thought the proof of the existence of the BLOCC protocols, one can consider entanglement as a quantity to which we can associate fluctuations, and whose fluctuations are constrained in much the same way as work is in the context of previous results in statistical mechanics. It is remarkable that the mathematical structure of thermodynamics and pure state entanglement transformation with a battery are so related, given the very different physical scenarios under consideration. For example, there is no heat bath in entanglement theory, the doubly-stochastic maps depend on the initial and final states of  particular transformation, unlike in thermodynamics, and the doubly stochastic maps take final states to initial states.

In classical thermodynamics (i.e. when the initial and final state are diagonal in the energy eigenbasis), it is unambiguous what the work is after a given process. In the \emph{implicit} case, one initially measures the total energy of system and bath, and performs the measurement once again at the end of the process. The energy difference must be the work which has been extracted. In the \emph{explicit} case where we include the battery and impose total energy conservation \cite{HO-limitations,alhambra2016fluctuating}, the work is just the energy difference in the battery before and after the transformation. Likewise, the system is originally in some particular energy level $\ket{i}$ and ends in some particular energy level $\ket{j}$. Thus the probability distribution over $i,j,w$ has a simple interpretation. 

However, in the quantum case, we cannot implement pure state thermodynamical transformations and expect that the work will always be a measurable quantity \cite{hayashi2017measurement,perarnau2017no}. This is because to implement arbitrary unitary transformations, one must have access to some system (the battery), which must be in a coherent state which is a superposition over many energy levels. Measuring the energy of this battery destroys the ability to implement the unitary transformation. We here see similar phenomena between the entanglement case, and the quantum thermodynamics case. We can measure the amount of entanglement consumed or extracted each time, but if we do so, then we cannot implement the transformation  $\ket{\psi}\rightarrow\ket{\phi}$. Nonetheless, the physical interpretation
of the fluctuation relation is clear, as it could in principle be measured.

In thermodynamics, work, whether classical or quantum, should be thought of as a process, not an observable \cite{alhambra2016fluctuating}. Classically, it is the process of raising or lowering a weight. Likewise, in the case of the entanglement battery, the entanglement fluctuation can be seen as the adding to or subtracting from the number of ebits in the battery. In both cases, the change in average quantities (whether work or entanglement) does not move the battery to an orthogonal state, and can only be measured on many copies. We thus have the intriguing phenomena that entanglement fluctuations in a battery  enable us to perform entanglement transformations, but do not correspond to a single von Neumann measurement. Crucially, the entanglement battery must have an uncertain amount of entanglement, that is, must not be in a state with a definite amount of entangled pairs, in order to aid in a pure state entanglement transformation. Since work can be identified with the process of raising and lowering the battery, it would be desirable to have a universal set of LOCC maps on the battery that applied to all the possible states in the system, as is the case in the thermodynamic scenario \cite{alhambra2016fluctuating}. This would perhaps amount to a more concrete description of the BLOCC protocols introduced here. While we believe this is possible, perhaps by starting from explicit constructions of LOCC protocols based on Nielsen's result \cite{Nielsen-pure-entanglement}, we leave the question for future work.

In the main section of this article, we have considered the case where the target state is only a single pure state $\ket{\phi}_{AB}\otimes \ket{\bat'}_{A'B'}$. In Section \ref{sec:ensembles} of the \supl \, we show that our results also hold in the case of an ensemble of pure target states $\left\{\ket{\phi^t}_{AB}\otimes \ket{\bat^t}_{A'B'}\right\}_t$. There, we take as an example the original entanglement concentration and dilution protocols. Using the results presented here, we can quantify the entanglement fluctuations in all concentration protocols, and we see why previous dilution protocols were sub-optimal. We show how to make them optimal and thus achieve true reversibility.

We find another interesting application of our results, in that we can solve a problem known as {\it partial entanglement recovery} \cite{morikoshi2000recovery}. There, one considers the irreversible LOCC transformation $\ket\psi_{AB}\rightarrow\ket\phi_{AB}$, and asks whether some of the entanglement can be recovered in the operation, by performing a transformation on an ancillary system  $\ket\phi_{AB}\otimes\chi_{A'B'}\rightarrow\ket\psi_{AB}\otimes\ket\omega_{A'B'}$. Since the choice of $\chi_{A'B'}$ and $\omega_{A'B'}$ are allowed to depend on $\psi_{AB}$ and $\phi_{AB}$, there is clearly the trivial solution where  $\chi_{A'B'}=\psi_{AB}$, and $\omega_{A'B'}= \psi_{AB}$, and one just performs the swap operation between $AA'$ and $BB'$. To rule out such trivial solutions, one can consider a notion of  {\it genuine  partial entanglement recovery}~\cite{bandyopadhyay2002partial} which requires limiting the dimension of the ancillary system. Progress on finding ways to partially recover some of the entanglement has been made in \cite{duan2006partial}. Here, we see that instead of restricting the dimension of the ancilla to enforce a notion of genuine recovery, we can instead demand that the ancillary state be universal. We then see that in fact, all the entanglement can be recovered.

Finally, the results here help us better understand the notion of catalysis and embezzlement, and provide a solution to the problem posed by approximate catalysis \cite{brandao2013second}. In exact catalysis, one asks whether there exists a system in state $\eta$ such that the transition $\rho\otimes\eta \rightarrow \sigma\otimes\eta$ is possible. In the case where the conditions for $\rho\otimes\eta \rightarrow \sigma\otimes\eta$ to be possible are given by majorization conditions, the conditions for $\rho\rightarrow\sigma$ to be possible catalytically were found to be given by the monotonicity of Renyi entropies~\cite{Klimesh-trumping,Turgut-trumping-2007}.

However, from a physical point of view, it is impossible to return a catalyst in exactly the same state, so it seems more natural to ask whether
$\rho\otimes\eta\rightarrow\sigma\otimes\eta'$ is possible, with $\eta\approx\eta'$. We thus see that our battery can be thought of as a type of catalyst in this sense. The conditions for $\rho\rightarrow\sigma$ to be possible under approximate catalysis depends on how close we should return the catalyst to its initial state \cite{brandao2013second}. In the case where we do not restrict the dimension of the catalyst, embezzlement \cite{Hayden-embezzling,harrow_entanglement_2009,leung2008coherent} poses a problem. Embezzling, is the process of extracting a resource from a state, without changing the state by very much. In the case of entanglement embezzling, we can extract ebits from an embezzling state without changing the state by much \cite{Hayden-embezzling}.

The problem then, is that if we allow approximate catalysis, any transformation is possible in the limit of an arbitrarily large catalyst, because we can use an embezzling state as an approximate catalyst. In fact, the amount by which the catalyst changes can be made arbitrarily small. This result has stalled entanglement theory, because one should allow approximate catalysis in any transformation, yet it seems to render the theory trivial, since all state transformations become possible. In the context of the present article, we see that we can bypass this issue. In particular, by quantifying how much fluctuating entanglement is transferred to the catalyst (in this case the entanglement battery), we can account for how much of it is used in any process. The fact that the catalyst's state changes very little does not take away from the fact that the amount of entanglement in it has changed by a significant amount. As a result, the theory does not trivialise.

Not only do we find an array of phenomena in entanglement akin to those found in thermodynamics, but also previous problems, such as entanglement recovery, the problem of embezzlement, and a strong converse of entanglement concentration, are related to those in thermodynamics. We thus see that many issues and open problems can be solved by connecting them to fluctuation relations. It is perhaps not surprising that fluctuation theorems for entanglement enable one to solve such problems, given the fertile research landscape that fluctuation theorems have opened up in the field of thermodynamics. Our hope is that likewise, fluctuating entanglement allows for the discovery of further phenomena in entanglement theory. Towards this aim, in Section \ref{sec:finitebattery} of the \supl \, we find that for the processing of a few qubits, an entanglement battery need not be large to be useful. This gives hope that experimental implementation of the protocols presented here, may one day be performed.

\alv{Most of the results here focus on the existence of LOCC protocols that implement the desired transitions, and do not specify the particular character or complexity of the unitaries and measurements. A similar problem occurs in the analogous thermodynamic scenario, for which recent results \cite{perry2018,lostaglio2018,baumer2018} show that one can in fact implement a large number of transitions with an experimentally feasible subset of free operations. This may be a starting point for a similar result in the present context of entanglement, which we leave for future work. }

\noindent
     {\bf Acknowledgments:} JO thanks Daniel Gottesman and Debbie Leung for interesting discussions, and an EPSRC Established Career Fellowship and Royal Society Wolfson merit award for support.
		CP acknowledges financial support from the European Research Council (ERC Grant Agreement no 337603), the Danish Council for Independent Research (Sapere Aude) and VILLUM FONDEN via the QMATH Centre of Excellence (Grant No. 10059).
		Research  at  Perimeter  Institute  is  supported  by  the
		Government of Canada through the Department of Innovation, Science and Economic
		Development and by the Province of Ontario through the Ministry of Research, Innovation and Science
		This work was supported in part by COST Action MP1209.


\clearpage
\widetext
\appendix

\section{Battery Assisted LOCC}

\subsection{The entanglement battery} \label{ss:battery}

In this section we expand the definition of the entanglement storage battery, and we show how to define it so that not only integer values of $w$ can be stored. In order to allow for the consumption or generation of non-integral amounts of entanglement, we substitute the product states and ebits in Eq.~\eqref{integer battery} with the following two types of states which contain almost the same amount of entanglement
\begin{eqnarray}
\label{G+}
\ket{\Gamma_u^+}
&=& 
\frac 1 {\sqrt u} \sum_{i=1}^u \ket{ii} \ ,
\\ \label{G-}
\ket{\Gamma^-_u}
&=& 
\frac 1 {\sqrt {u-1}} \sum_{i=u+1}^{2u-1} \ket{ii} \ .
\end{eqnarray}
The state $\ket{\Gamma^+_u}$ contains $\log u$ e-bits of entanglement, while $\ket {\Gamma^- _u}$ contains $\log (u-1)$ e-bits. Hence, in going from one to the other $\ket{\Gamma^+_u} \to \ket{\Gamma^- _u}$, the amount of entanglement that we borrow is 
\begin{equation}
\delta w = \log\left(\frac{u}{u-1}\right) 
\approx \frac 1 u\ , 
\end{equation}  
where the above approximation holds in the large $u$ limit. Therefore, the parameter $u$ controls how fine-grained the entanglement scale is. We will henceforth choose $u$ large enough such that all values of $w$ are as close as required to multiples of $\delta w$. Also note that the two states $\ket{\Gamma^+_u}, \ket{\Gamma^- _u}$ are locally distinguishable.

The states of the battery that have a well-defined amount of entanglement are the following
\begin{equation}
\ket{e_x}_{A'B'}=
\underbrace{
\ket{\Gamma^+_u}\otimes \cdots \otimes \ket{\Gamma^+_u}
}_x
\otimes \underbrace{
\ket{\Gamma^-_u}\otimes \cdots \otimes \ket{\Gamma^-_u}
}_{n-x}\ ,
\end{equation}
for all integers $x\in \left\{0,\dots,n\right\}$ and $A'$,$B'$ labeling Alice and Bob's halves of the battery respectively. In the protocols that we consider, all battery states are contained in the subspace generated by $\left\{\ket {e_x}\right\}_{x=0}^{n}$. The reduced state on Alice or Bob's half of the battery is then:
\begin{equation}
s_x= \tr_{B'} \proj{e_x} =
\left(\frac{1}{u}
\sum_{i=1}^{u}\ketbra{i}{i}\right)^{\otimes x} \otimes \left(\frac{1}{u-1}
\sum_{i=u+1}^{2u-1}\ketbra{i}{i}\right)^{\otimes\left(n-x\right)}.
\end{equation}
The set of states $\left\{s_x\right\}_{x=0}^{n}$ are also orthogonal and live on some subspace $A'$ of $\mathcal{H}^{\otimes n}$ where $\mathcal{H}=\mathbb{C}^{2u-1}$. When restricted to Alice's system, our LOCC protocol will map this subspace to itself. We will often use a suitable restriction of $\left\{\ket{z}\right\}_{z\in\left\{1,2,\dots,2u-1\right\}^n}$ as an orthonormal basis for $A'$. We will write $z \in s_x$ to denote that $\ket{z}$ belongs to the support of $s_x$.
Note also that:
\begin{equation}
\sum_{x=0}^{n} u^x\left(u-1\right)^{n-x} s_x = \mathbb{I}_{A'},
\end{equation}
so the orthogonal projectors $P_x = u^x\left(u-1\right)^{n-x} s_x$ give a resolution of the identity on $A'$. In general, the initial state of the battery is denoted by Eq.~\eqref{eq:battery}.

\subsection{Reversible pure-state transformations require a battery in a uniform superposition}
\label{ss:uniform}

In this Section of the \supl\ we show that the state of the battery $\ket \eta _{A'B'} = \sum_x \gamma_x \ket {e_x}_{A'B'}$ must be close to a uniform superposition of entanglement eigenstates $\ket{e_x}_{A'B'}$, if we assume the following two conditions:
\begin{itemize}
\item The only allowed actions on the battery are rising and lowering the amount of entanglement, so that the final state is of the form
\begin{equation}
\label{general final state}
	\ket {\Phi}_{ABA'B'} =
	\sum_r \ket {\phi_r} _{AB} \otimes
	\Delta^r \ket \eta _{A'B'} \ .
\end{equation}
  
\item The state of the battery $\ket \eta_{A'B'}$ allows for approximately implementing all reversible pure-to-pure entanglement transformations $\ket \psi_{AB} \to \ket \phi_{AB}$. That is, for every $\epsilon>0$, there is a reversible BLOCC transformation with final state~\eqref{general final state} being $\epsilon$-close to the target one
\begin{equation}
	\label{targer}
	\|\Phi_{ABA'B'} 
	-\phi_{AB} \otimes\eta_{A'B'} \|_1
	\leq \epsilon\ ,
\end{equation}
and with identical marginal
\begin{equation}
  \label{equal marginals}
  \Phi_B = \phi_B\ .
\end{equation}

\end{itemize}

The first condition is what allows us to quantify the notion of an entanglement fluctuation, by defining it to be the adding or subtracting of the number of ebits of the battery. Before proving the uniformity of $\ket\eta$ let us collect some useful facts. The non-zero Schmidt coefficients of $\ket {e_x}$ are 
\begin{equation}
  \mathcal S _{\ket{e_x}} =
  \left\{ 
  \xi_x^{-1/2} ,
  \mbox{ appearing } \xi_x   
  \mbox{ times}
  \right\}\ ,
\end{equation}
where we define the constants
\begin{equation}
  \xi_x = 
  (u-1)^{n} 
  \left(\frac u {u-1} \right)
  ^{x}\ .  
\end{equation}
For any components $\gamma_x$, the non-zero Schmidt coefficients of $\ket \eta = \sum_x \gamma_x \ket {e_x}$ are 
\begin{equation}
  \mathcal S _{\ket\eta} =
  \left\{ |\gamma_x|\xi_x^{-1/2},
  \mbox{ appearing } \xi_x   
  \mbox{ times, for all } x
  \right\}\ .
\end{equation}
Recalling that $\Delta^{\delta w} \ket{e_x} = \ket{e_{x+1}}$ and $\delta w = \log(u/(u-1))$, we arrive at 
\begin{equation}
  \label{cart prod}
  \mathcal S_{\Delta^{\delta w} \ket\eta} = 
  \left\{ 
  |\gamma_x|\xi_x^{-1/2} 
  \sqrt{\frac {u-1} u},
  \mbox{ appearing } 
  \xi_x \frac u {u-1}  
  \mbox{ times, for all } x
  \right\}\ .
\end{equation}
Also, we note that without loss of generality we can assume that the coefficients $\gamma_x$ are real and positive. Hence, we define $\alpha_x = \gamma_x^2$, which satisfy normalization $\sum_x \alpha_x =1$.

Now, let us consider the particular pure-to-pure reversible transformation $\ket \psi_{AB} \to \ket \phi_{AB}$ with
\begin{eqnarray}
  \label{A psi}
  \ket \psi _{AB} &=& 
  \frac 1 {\sqrt{2}} \left(
  \ket{0,0}_{AB} 
  + \ket{\Gamma_u^+}_{AB} \right)
  \ , \\ \label{A phi}
  \ket \phi _{AB} &=& 
  \frac 1 {\sqrt{2}} \left(
  \ket{0,0}_{AB} 
  + \ket{\Gamma_u^-}_{AB} \right)
  \ ,
\end{eqnarray}
where $\ket{\Gamma_u^\pm}_{AB}$ are defined in~(\ref{G+}-\ref{G-}). The states $\{\ket 0_A , \ket 1_A, \ldots, \ket {2u-1} _A\}$ form an orthonormal basis for Alice's Hilbert space, and analogously for Bob. It is known~\cite{Nielsen-pure-entanglement} that reversibility is only possible when the Schmidt coefficients of the initial and final states are identical. And, since the Schmidt coefficients of the two states~\eqref{A psi} and~\eqref{A phi} are different, reversibility can only be achieved with a non-trivial action on the battery. Let us prove that, if the global initial state is $\ket \Psi _{ABA'B'} = \ket{\psi}_{AB} \otimes \ket \eta _{A'B'}$, then the global final state must be
\begin{equation}
  \label{gen fin s}
  \ket \Phi _{ABA'B'}
  =
  \frac 1 {\sqrt 2} \left(
  \ket{0,0}_{AB} 
  \otimes \ket \eta _{A'B'}
  + \ket{\Gamma_u^-}_{AB} 
  \otimes \Delta^{\delta w} \ket \eta _{A'B'} 
  \right)  \ .
\end{equation}
The Schmidt coefficients of the initial state are the Cartesian product 
\begin{equation}
  \label{target SC}
  \bigg\{\frac 1 {\sqrt 2} , \underbrace{\frac 1 {\sqrt{2u}}, \ldots, \frac 1 {\sqrt{2u}}}_u  \bigg\}
  \times \mathcal S_{\ket\eta}\ .
\end{equation}
Now, we must show that the only final state of the form~\eqref{general final state} with the above Schmidt coefficients is~\eqref{gen fin s}. Invoking~\eqref{equal marginals} we obtain 
\begin{equation}
  \label{A12}
  \ket \Phi_{ABA'B'}
  =
  \frac 1 {\sqrt 2} \ket{0,0}_{AB} 
  \otimes \ket {\eta_0} _{A'B'}
  + \frac 1 {\sqrt{2(u-1)}} \sum_{j=1}^{u-1} 
  \ket{j,j}_{AB} 
  \otimes \ket {\eta_j} _{A'B'} 
  \ ,
\end{equation}
with
\begin{equation}
  \ket {\eta_j} = 
  \sum_r c_j^r \Delta^r \ket \eta
  \ ,
\end{equation}
and $j = 1,2, \ldots, u-1$.
If there is a value of $j$ with more than one non-zero $c_j^r$, then, the Schmidt coefficients $|\sum_r c_j^r \sqrt{\alpha_{x-r}}|\, \xi_x^{-1/2}$ cannot be of the form $\sqrt{\alpha_{x}}\, \xi_x^{-1/2}$. Because for any $x \geq n- \frac {r_{\max}} {\delta w}$ we must have $\alpha_x=0$, where $r_{\max}$ denotes the largest value of $r$ in this transformation. Therefore, state~\eqref{A12} can also be written as
\begin{equation}
  \ket \Phi_{ABA'B'}
  =
  \frac 1 {\sqrt 2} \ket{0,0}_{AB} 
  \otimes \Delta^{r_0}\ket {\eta} _{A'B'}
  + \frac 1 {\sqrt{2(u-1)}} \sum_{j=1}^{u-1} 
  \ket{j,j}_{AB} 
  \otimes \Delta^{r_j} \ket {\eta} _{A'B'} 
  \ .
\end{equation}
Now, using~\eqref{cart prod}, we see that the only way to recover the Schmidt coefficients~\eqref{target SC} is to set $r_j=\delta w$ for all $j>0$ and $r_0=0$. This is precisely~\eqref{gen fin s}.

Next, we prove the uniformity of the coefficients $\alpha_x$ by invoking condition~\eqref{targer}. In order to do this, we need to compute the partial trace of~\eqref{gen fin s}, which is
\begin{equation}
  \Phi _{AB}
  =
  \frac 1 2 \Big(
  \proj{0,0}
  + \proj{\Gamma_u^-}
  + \ket{0,0}\! \bra{\Gamma_u^-}
  \bra\eta \Delta^{-\delta w} \ket\eta
  + \ket{\Gamma_u^-}\! \bra{0,0} 
  \bra\eta \Delta^{\delta w} \ket \eta
  \Big)  \ .
\end{equation}
Substituting this in~\eqref{targer} we obtain
\begin{equation}
  \frac 1 2 \Big\|
  \ket{0,0}\! \bra{\Gamma_u^-}
  \big(\bra\eta \Delta^{-\delta w} \ket\eta -1\Big)
  + \ket{\Gamma_u^-}\! \bra{0,0} 
  \Big( \bra\eta \Delta^{\delta w} \ket \eta -1\Big)
  \Big\|_1 \leq \epsilon\ ,
\end{equation}
which is equivalent to
\begin{equation}
  \sum_x \sqrt{\alpha_x \alpha_{x+1}}
  =
  \bra\eta \Delta^{\delta w} \ket \eta 
  \geq 1-\epsilon\ .
\end{equation}
Using the identity 
\begin{equation}
  \frac 1 2 \big\| \proj \psi -\proj \phi \big\|_1
  =
  \sqrt{1- \left|\langle \psi \ket \phi \right|^2}
\end{equation}
on the two pure states $\ket \eta$ and $\Delta^{\delta w} \ket \eta$, we obtain
\begin{equation}
  \label{dist fid}
  \sum_x \left|
  \alpha_x -\alpha_{x+1} \right| 
  \leq 
  2 \sqrt{1- \left( \mbox{$\sum_x$} \sqrt{\alpha_x \alpha_{x+1}} \right)^2}
  \leq
  \sqrt{8\epsilon}\ .
\end{equation}
And finally, applying the triangular inequality, we obtain
\begin{equation}
  \label{s dist alpha}
  \sum_x \left|\alpha_x -\alpha_{x+y} \right| 
  \leq 
  y \sqrt{8\epsilon}
  \ ,
\end{equation}
for all $y$ such that $|y| \leq r_{\max}/\delta w$.

\alv{It is clear that, given $\epsilon$ and $y$, the inequality will hold the more uniform the set of $\alpha_x$ is. For instance, if we take $\alpha_x=1/N$ uniform, we have that $\sum_x \left|\alpha_x -\alpha_{x+y} \right| =\frac{2y}{N}$, so an error of $\epsilon$ in the transformation means we need a width of $N\ge \frac{1}{\sqrt{2 \epsilon}}$.}

\newpage
\section{Necessary conditions for pure-state BLOCC} \label{sec:necc}

In this section we prove the necessary part of Result \ref{result:bat-stochastic}. That is, if a BLOCC protocol exists for the transformation given in Eq.~\eqref{eq:goal}, then there must exist a conditional distribution $P\left(i,w|j\right)$ satisfying the three conditions given by Eqs.~\eqref{eq:normalisation}-\eqref{eq:doesthetrick}. Throughout, we shall denote the density matrix of a pure state $\ket{\psi}_{AB}$ by $\psi_{AB}$ and its reduced state on subsystem $A$ by $\psi_A$, omitting subsystem labels when the context is clear.

\alv{We divide the proof into two parts. First, we outline the structure of the LOCC operations we consider. We then use this structure to prove that the existence of a BLOCC protocol implies the existence of a stochastic matrix that obeys the constraints of Result \ref{result:bat-stochastic}.}

\subsection{Pure-state LOCC transformations}

Let $\Lambda$ denote the CP-map associated to a particular BLOCC protocol.
Since this transforms pure states to pure states (on system plus battery) we can assume \cite{LoPopsecu1997-beyond} that $\Lambda$ consists of the following steps:
\begin{enumerate}
\item Alice performs a POVM $\left\{M_m\right\}$ on $AA'$.
\item Alice sends the outcome $m$ to Bob.
\item Bob applies a correction unitary $V_m$ on $BB'$.
\end{enumerate}
Thus, following \cite{Nielsen-pure-entanglement}, we have
\begin{equation}
  \Phi_{ABA'B'} =
  \Lambda\left(\Psi_{ABA'B'}\right) = 
  \sum_m \left(M_m\otimes V_m\right)\Psi_{ABA'B'}
    \left(M_m\otimes V_m \right)^\dagger.
\end{equation}
Imposing the purity of the final state $\Phi_{ABA'B'}$, we get
\begin{equation}
  \Phi_{ABA'B'}
  \propto 
  \left(M_m\otimes V_m\right)  
  \Psi_{ABA'B'}\left(M_m\otimes V_m \right)^\dagger,\quad \forall m\ ,
\end{equation}
which implies the existence of some positive coefficients $r_m$ satisfying
\begin{equation} \label{eq:forwardp1}
  \left(M_m \otimes V_m\right) 
  \ket{\Psi}_{ABA'B'} 
  =
  \sqrt {r_m}\,  
  \ket{\Phi}_{ABA'B'}\ .
\end{equation}
Applying the unitary $V_m^\dagger$ on the two sides we get $(M_m \otimes \id) \ket{\Psi} = \sqrt {r_m}\, (\id\otimes V_m^\dagger) \ket{\Phi}$ and
\begin{equation}
  \bra {\Psi}(M_m^\dagger M_m) \otimes X \ket{\Psi}
  = r_m
  \bra {\Phi} \id\otimes (V_m XV_m^\dagger) \ket{\Phi}
  \ ,
\end{equation}  
for any operator $X$, where we have omitted the specification of subsystems $ABA'B'$.
Using $\sum_m M_m^\dagger M_m = \id$, we obtain
\begin{equation}
  \label{X eq}
  \bra {\Psi} \id\otimes X \ket{\Psi} 
  =
  \bra {\Phi} \id\otimes \Lambda_{BB'} (X) \ket{\Phi}
  \ ,
\end{equation}  
where $\Lambda_{BB'}$ is 
\begin{equation}
  \Lambda_{BB'} (X) = 
  \sum_m r_m V_m X V_m^\dagger\ .
\end{equation}
We emphasize that $\Lambda_{BB'}$ depends on the initial state $\ket{\Psi}_{ABA'B'}$ via the probabilities $r_m = \bra {\Psi} (M_m^\dagger M_m) \otimes \id \ket{\Psi}$.
This allows us to relate the map $\Lambda_{BB'}$ to the global map generated by the actual protocol
\begin{equation}
  \Lambda_{BB'} (X) =
  \tr_{AA'}\! \left[
  \Lambda_{ABA'B'} (\Psi_{AA'} \otimes X)
  \right]\ .
\end{equation}
Eq.~\eqref{X eq} can also be written as
\begin{equation}
  \label{eq central}
  \Psi_{BB'} = \Lambda_{BB'}^* (\Phi_{BB'})\ ,
\end{equation}
where $\Psi_{BB'}$ and $\Phi_{BB'}$ are Bob's initial and final reduced states, and $\Lambda_{BB'}^*$ is the dual map of $\Lambda_{BB'}$.

\subsection{Necessary conditions for \texorpdfstring{$P(i,w|j)$}{P(i,w|j)}}

In what follows we define the conditional distribution $P(i,w|j)$, which captures relevant information about the CP-map of the BLOCC transformation $\Lambda$.
In order to derive the necessary conditions we need to define $P(i,w|j)$ imposing that the final system-battery state is product $\Phi_{ABA'B'} = \phi_{AB} \otimes\eta_{A'B'}$, which, as expressed in Eq.~\eqref{targer}, is true in the limit $\epsilon\to 0$. Hence,
\begin{eqnarray}
  P(i,w|j) &=&
  \sum_{x'} \tr\! \left[
  (\proj i \otimes P_{x' - \frac w {\delta w}})
  \Lambda_{BB'}^* \left(
  \proj j \otimes \left[
  P_{x'}\, \eta_{B'}\, P_{x'}
  \right]\right)\right] \label{eq:relationP}
  \\ &=& 
  \sum_{x'} \alpha_{x'} \tr\! \left[
  (\proj i \otimes P_{x'- \frac w {\delta w}})
  \Lambda_{BB'}^*(\proj j \otimes s_{x'})
  \right] \label{eq:relationP2},
\end{eqnarray}
which corresponds to the statistics obtained in the following 5-step procedure:
\begin{enumerate}
  \item Prepare the state $\ketbra{j}{j}_B \otimes \eta_{B'}$.
  
  \item Measure the position of the battery $P_{x'}$.
  
  \item Transform the resulting state with the map $\Lambda_{BB'}^*$.
  
  \item Measure the system with $\proj i$ and the battery with $P_{x}$.
  
  \item Record the variable $w= (x'-x) \delta w$ and forget $x$ and $x'$.
\end{enumerate}

Let us now see that Eqs~\eqref{eq:normalisation} - \eqref{eq:doesthetrick} are necessary.
By construction, $P(i,w|j)$ is a normalised probability distribution
\begin{eqnarray}
  \nonumber
  \sum_{i,w} P(i,w|j) 
  &=& 
  \sum_{x'} \tr\! \left[
  (\id \otimes \id)
  \Lambda_{BB'}^* \left(
  \proj j \otimes \left[
  P_{x'}\, \eta_{B'}\, P_{x'}
  \right]\right)\right]
  \\ &=&
  \tr\! \left[
  \Lambda_{BB'}^* \left(
  \proj j \otimes \eta_{B'}
  \right)\right]
  = 1 ,
\end{eqnarray}
where we have used that the map $X \mapsto \sum_{x'} P_{x'} X P_{x'}$ is trace-preserving. Hence, we have shown that Eq.~\eqref{eq:normalisation} holds.

Now, let us move on to proving Eq.~\eqref{eq:w-stochasticity}.
Using $P_x = u^x \left( u-1 \right)^{n-x} s_x $ and the unitality of the map $\Lambda_{BB'}^*$, we obtain
\begin{eqnarray}
  \nonumber
  \sum_{w,j} P(i,w|j) 2^{w}
  &=&
  \sum_{w,x'} \alpha_{x'} \tr\! \left[
  (\proj i \otimes s_{x' -\frac{w}{\delta w}})
  \Lambda_{BB'}^*(\id \otimes P_{x'})
  \right]
  \\ \nonumber &\approx &
  \sum_{w,x'} \alpha_{x'-\frac{w}{\delta w}} 
  \tr\! \left[
  (\proj i \otimes s_{x' -\frac{w}{\delta w}})
  \Lambda_{BB'}^*(\id \otimes P_{x'})
  \right]
  \\ &= &
  \tr\! \left[
  (\proj i \otimes \eta_{B'})
  \Lambda_{BB'}^*(\id \otimes \id)
  \right]
  = 1\ ,
\end{eqnarray}
where we have approximated $\alpha_{x'} \approx \alpha_{x' -\frac{w}{\delta w}}$.
We can bound the accuracy of this approximation as
\begin{eqnarray}
  \nonumber
  && \left| 
  \sum_{w,j} P(i,w|j)\, 2^{w}
  -1 \right|
  \\ \nonumber &=&
  \left| \sum_{w,x'} 
  (\alpha_{x'} -\alpha_{x' -\frac{w}{\delta w}})\, 
  \tr\! \left[
  (\proj i \otimes s_{x' -\frac{w}{\delta w}})
  \Lambda_{BB'}^*(\id \otimes P_{x'})
  \right]\right| 
  \\ \nonumber &=&
  \left| \sum_{w,x} 
  (\alpha_{x'+\frac{w}{\delta w}} -\alpha_{x'} )\, 
  \tr\! \left[
  (\proj i \otimes s_{x'})
  \Lambda_{BB'}^*(\id \otimes P_{x'+\frac{w}{\delta w}})
  \right]\right| 
  \\ \nonumber &\leq &
  \sum_{w:|w|\leq w_{\max}}\sum_{x'} 
  \left| \alpha_{x'+\frac{w}{\delta w}} -\alpha_{x'} \right|
  \\ &\leq &
  \sqrt{8\, \epsilon}\  
  \frac {w_{\max}} {\delta w} 
  \approx 
  \sqrt{8\, \epsilon}\ w_{\max}\, u\ ,
\end{eqnarray}
where we have used Eq.~\eqref{s dist alpha}, and the fact that  the number of values of $w$ in the range $|w| \leq w_{\max}$ and with discretisation $\delta w$ is approximately $w_{\max} /\delta w$.
That is, for fixed $w_{\max}$ and $u$, the approximation becomes more exact as $\epsilon$ tends to zero. We thus see that $P\left(i,w|j\right)$ satisfies Eq.~\eqref{eq:w-stochasticity} in the limit $\epsilon\rightarrow 0$.

To obtain Eq.~\eqref{eq:doesthetrick}
we use Eq.~\eqref{eq central} and the approximate equality of reduced states $\Phi_{BB'} \approx \phi_B \otimes \eta_{B'} = (\sum_{j} q_j \proj j) \otimes (\sum_{x'} \alpha_{x'} s_{x'})$ in the following
\begin{eqnarray}
  \nonumber
  \sum_{w,j} P(i,w|j) q_j 
  &=&
  \sum_{j,x'} q_j\, \alpha_{x'} \tr\! \left[
  (\proj i \otimes \id)
  \Lambda_{BB'}^*(\proj j \otimes s_{x'})
  \right]
  \\ \nonumber &=& 
  \tr\! \left[(\proj i \otimes \id)
  \Lambda_{BB'}^*(\phi_B \otimes \eta_{B'}) \right]
  \\ \nonumber &\approx &
  \tr\! \left[(\proj i \otimes \id)
  \Lambda_{BB'}^*(\Phi_{BB'}) \right]
  \\ &=& 
  \tr\! \left[(\proj i \otimes \id)
  \Psi_{BB'}  \right]
  = p_i\ .
\end{eqnarray}
Now, we can bound the accuracy of the above approximation by using assumption~\eqref{targer} as
\begin{eqnarray}
  \nonumber
  \sum_i \left| \sum_{w,j} P(i,w|j) q_j 
  -p_i\right|
  &=&
  \sum_i \left|
  \tr\! \left[(\proj i \otimes \id)
  \Lambda_{BB'}^*(\phi_B \otimes \eta_{B'} -\Phi_{BB'}) \right]
  \right|
  \\ \nonumber &\leq & 
  \left\|
  \Lambda_{BB'}^*(\phi_B \otimes \eta_{B'} -\Phi_{BB'})
  \right\|_1
  \\ &\leq & 
  \left\| \phi_B \otimes \eta_{B'} -\Phi_{BB'} \right\|_1
  \leq \epsilon\ ,
\end{eqnarray}
where we have used that $\|X\|_1 = |\max_{0\leq P \leq \id} \tr PX |$ and that $\Lambda_{BB'}^*$ is a trace-preserving CP-map.

\section{Sufficient conditions for pure-state BLOCC} \label{sec:suff}

In this section we prove the sufficient part of Result \ref{result:bat-stochastic}. That is, if there exists a conditional distribution $P\left(i,w|j\right)$ satisfying the three conditions given by Eqs.~\eqref{eq:normalisation}-\eqref{eq:doesthetrick}, then there exists a BLOCC protocol for the transformation given in Eq.~\eqref{eq:goal}.

We start with a conditional probability distribution $P\left(i,w|j\right)$ satisfying Eq.~\eqref{eq:normalisation}, \eqref{eq:w-stochasticity} and \eqref{eq:doesthetrick}. Our goal is to show that given such a probability distribution, it is possible to construct a Battery Assisted LOCC protocol that converts a bipartite pure state with Schmidt coefficients $p_i$ into a bipartite pure state with coefficients $q_j$ whilst extracting a coherent entanglement distribution $\left\{P\left(w\right),w\right\}$.

To do this, we first need to pick a battery of size $n$ and with spacing parameter $u$ such that it is capable of incorporating fluctuations by $w$. In other words, we want to pick $u$ such that for each $w$ there exists an integer $a^{u}_{w}$ such that $w\approx a^{u}_{w} \log\left(\frac{u}{u-1}\right)$. As $u$ increases, this approximation improves. More specifically, for fixed $u$, we take $a^{u}_{w}$ to be the greatest integer such that:
\begin{equation}
a^{u}_{w} \log\left(\frac{u}{u-1}\right) \leq w.
\end{equation}
With respect to this, $P\left(i,w|j\right)$ satisfies (using Eq.~\eqref{eq:w-stochasticity}):
\begin{equation} \label{eq:gen gibbs approx}
\sum_{j,w} P\left(i,w|j\right)\left(\frac{u}{u-1}\right)^{a^{u}_{w}} \leq 1.
\end{equation}
Introducing $a^{u}_{w}$ allows us to deal with the fact that general $w$ cannot be written as an exact multiple of $\log\left(\frac{u}{u-1}\right)$. We will define the maximum of all these as
 \begin{equation}\label{eq:maxau}
 a^{u}_{\textrm{max}}=\max_w\left\{\left|a^{u}_{w}\right|\right\}.
 \end{equation}

To prove sufficiency, we will first construct a series of LOCC protocols $\Lambda_N$, indexed by $N:=n-2a^{u}_{\textrm{max}}$, such that (for even $N$):
\begin{equation}
  \ket{\Psi_N}=\sum_{i=1}^{d}\sum_{x=0}^{n} \sqrt{\sum_w \frac{p_{iw}}{N+1} \delta_{x+a^{u}_{w}\in\left\{\frac{n-N}{2},\dots,\frac{n+N}{2}\right\}}} \ket{ii}\otimes \ket{e_x},
  \label{eq:initialN}
\end{equation}
where $p_{iw}:=\sum_{j=1}^{d}P\left(i,w|j\right)q_j$, is converted into:
\begin{equation}
  \label{eq:finalN}
\ket{\Phi_N}=\sum_{j=1}^{d}\sum_{x'=0}^{n}  \sqrt{\frac{q_j}{N+1}}\delta_{x'\in\left\{\frac{n-N}{2},\dots,\frac{n+N}{2}\right\}} \ket{jj}\otimes\ket{e_{x'}}.
\end{equation}
Note that here we take an initial state that is correlated between system and battery and convert it into a final state which is product across this divide and has support of size $N+1$ on the battery. However, the protocol can also be applied in the case where the initial system and battery are uncorrelated. This we consider in Section \ref{ssec:large n} of the \supl, where we prove that in the limit of large $N$, the state in Equation \eqref{eq:initialN} tends to a product state, and thus acting the protocol on an initial product state will result in a target state arbitrarily close to the ideal one of Equation \eqref{eq:finalN} and an entanglement distribution which is also arbitrarily close to the ideal one.

\subsection{Construction of \texorpdfstring{$\Lambda_N$}{Lambda}}
\label{sec:constructionlambda}
To show the existence of a protocol converting $\ket{\Psi_N}$ into $\ket{\Phi_N}$, we ultimately need to construct a doubly stochastic matrix that maps the Schmidt coefficients of $\ket{\Phi_N}$ to those of $\ket{\Psi_N}$ \cite{Nielsen-pure-entanglement}. We do this in three steps.

\subsubsection{Conversion to \texorpdfstring{$P\left(i,x|j,x'\right)$}{P(i,x|j,x')}}

\noindent From $P\left(i,w|j\right)$ and $a^{u}_{w}$, we first define the object $P\left(i,x|j,x'\right)$ via:
\begin{equation}\label{eq:pix}
P\left(i,x|j,x'\right)=\sum_w P\left(i,w|j\right) \delta_{x'-x,a^u_w},
\end{equation}
where $x$ and $x'$ are integers between $\pm\infty$. Next, we rewrite Eqs.~\eqref{eq:normalisation}, \eqref{eq:gen gibbs approx} and \eqref{eq:doesthetrick} in terms of this new object.

\noindent Using Eq.~\eqref{eq:normalisation}, we see that $P\left(i,x|j,x'\right)$ satisfies:
\begin{align*}
\sum_{i=1}^{d}\sum_{x=-\infty}^{\infty}P\left(i,x|j,x'\right)&=\sum_{i=1}^{d}\sum_{x=-\infty}^{\infty}\sum_w P\left(i,w|j\right)\delta_{x'-x,a^u_w}\\
&=\sum_{i=1}^{d}\sum_w P\left(i,w|j\right)\\
&=1,
\end{align*}
while Eq.~\eqref{eq:gen gibbs approx} gives that:
\begin{align*} 
\sum_{j=1}^{d}\sum_{x'=-\infty}^{\infty}P\left(i,x|j,x'\right)\left(\frac{u}{u-1}\right)^{x'-x}&=\sum_{j=1}^{d}\sum_{x'=-\infty}^{\infty}\sum_w P\left(i,w|j\right)\delta_{x'-x,a^u_w}\left(\frac{u}{u-1}\right)^{x'-x}\\
&=\sum_{j=1}^{d}\sum_w P\left(i,w|j\right)\left(\frac{u}{u-1}\right)^{a^u_w}\\
&\leq 1.
\end{align*}
Finally, Eq.~\eqref{eq:doesthetrick} can be used to show that:
\begin{align*} 
&\sum_{j=1}^{d}\sum_{x'=-\infty}^{\infty}P\left(i,x|j,x'\right)\frac{q_j}{N+1}\delta_{x'\in\left\{\frac{n-N}{2},\dots,\frac{n+N}{2}\right\}}\\
=&\sum_{j=1}^{d}\sum_{x'=-\infty}^{\infty}\sum_w P\left(i,w|j\right)\delta_{x'-x,a^u_w}\delta_{x'\in\left\{\frac{n-N}{2},\dots,\frac{n+N}{2}\right\}}\frac{q_j}{N+1}\\
=&\sum_{j=1}^{d}\sum_w P\left(i,w|j\right)\frac{q_j}{N+1}\delta_{x+a^u_w\in\left\{\frac{n-N}{2},\dots,\frac{n+N}{2}\right\}}\\
=&\sum_w \frac{p_{iw}}{N+1} \delta_{x+a^{u}_{w}\in\left\{\frac{n-N}{2},\dots,\frac{n+N}{2}\right\}}.
\end{align*}

To summarise, our three equations are now:
\begin{align}
\sum_{i=1}^{d}\sum_{x=-\infty}^{\infty}P\left(i,x|j,x'\right)&=1, \label{eq:P stoch}\\
\sum_{j=1}^{d}\sum_{x'=-\infty}^{\infty}P\left(i,x|j,x'\right)\left(\frac{u}{u-1}\right)^{x'-x}&\leq1, \label{eq:P gen gibbs approx}\\
\sum_{j=1}^{d}\sum_{x'=-\infty}^{\infty}P\left(i,x|j,x'\right)\frac{q_j}{N+1}\delta_{x'\in\left\{\frac{n-N}{2},\dots,\frac{n+N}{2}\right\}}&=\sum_w \frac{p_{iw}}{N+1} \delta_{x+a^{u}_{w}\in\left\{\frac{n-N}{2},\dots,\frac{n+N}{2}\right\}}. \label{eq:P trans}
\end{align}
Note also that in a refinement of Eq.~\eqref{eq:P stoch}, for $x'\in\left\{\frac{n-N}{2},\dots,\frac{n+N}{2}\right\}$ we have that:
\begin{equation} \label{eq:P stoch exact}
\sum_{i=1}^{d}\sum_{x=0}^{n} P\left(i,x|j,x'\right)=1.
\end{equation}

\subsubsection{Construction of a doubly sub-stochastic matrix}

From $P\left(i,x|j,x'\right)$ we will now construct a matrix with rows and columns labeled by the Schmidt basis of system-battery, $\ket{i,z}$ and $\ket{j,z'}$ respectively. This matrix will be doubly sub-stochastic (the row and column sums will be less than or equal to one) but it will have the important property of mapping the Schmidt coefficients of $\ket{\Phi_N}$ to those of $\ket{\Psi_N}$.

Define for all $z\in s_x$, $z'\in s_{x'}$ where $x,x'\in\left\{0,\dots,n\right\}$:
\begin{equation}
R\left(i,z|j,z'\right)=P\left(i,x|j,x'\right) u^{-x}\left(u-1\right)^{x-n}.
\end{equation}
$R\left(i,z|j,z'\right)$ is a square, doubly sub-stochastic matrix. To see this note that if we had not truncated the range of $x$ and $x'$ to lie in $\left\{0,\dots,n\right\}$ and assumed that the degeneracy of $z\in s_y$ was $u^{y}\left(u-1\right)^{n-y}$ (regardless of the fact that this does not make much sense for $y<0$ or $y>n$) we would have had from Eq.~\eqref{eq:P stoch} that:
\begin{align*}
\sum_{i=1}^{d}\sum_{z} R'\left(i,z|j,z'\right)&=\sum_{i=1}^{d}\sum_{x=-\infty}^{\infty}\sum_{z\in s_x} P\left(i,x|j,x'\right) u^{-x}\left(u-1\right)^{x-n}\\
&=\sum_{i=1}^{d}\sum_{x=-\infty}^{\infty}P\left(i,x|j,x'\right)\\
&=1
\end{align*}
and using Eq.~\eqref{eq:P gen gibbs approx} that:
\begin{align*}
\sum_{j=1}^{d}\sum_{z'}R'\left(i,z|j,z'\right)&=\sum_{j=1}^{d}\sum_{x'=-\infty}^{\infty}\sum_{z'\in s_{x'}} P\left(i,x|j,x'\right) u^{-x}\left(u-1\right)^{x-n}\\
&=\sum_{j=1}^{d}\sum_{x'=-\infty}^{\infty} P\left(i,x|j,x'\right) \left(\frac{u}{u-1}\right)^{x'-x}\\
&\leq1
\end{align*}
where $R'$ is the non-truncated version of $R$. 

While $R$ is not doubly stochastic, it does satisfy (using Eq.~\eqref{eq:P trans}):
\begin{align*}
&\sum_{j=1}^{d}\sum_{x'=\frac{n-N}{2}-a^{u}_{\textrm{max}}}^{\frac{n+N}{2}+a^{u}_{\textrm{max}}}\sum_{z\in s_{x'}} R\left(i,z|j,z'\right) \frac{q_j}{N+1} u^{-x'}\left(u-1\right)^{x'-n} \delta_{x'\in\left\{\frac{n-N}{2},\dots,\frac{n+N}{2}\right\}}\\
=&\sum_{j=1}^{d} \sum_{x'=-\infty}^{\infty}P\left(i,x|j,x'\right)\frac{q_j}{N+1} u^{-x}\left(u-1\right)^{x-n}\delta_{x'\in\left\{\frac{n-N}{2},\dots,\frac{n+N}{2}\right\}}\\
=&\sum_w \frac{p_{iw}}{N+1} u^{-x}\left(u-1\right)^{x-n} \delta_{x+a^{u}_{w}\in\left\{\frac{n-N}{2},\dots,\frac{n+N}{2}\right\}},
\end{align*}
i.e. it maps the Schmidt coefficients of $\ket{\Phi_N}$ to those of $\ket{\Psi_N}$.

Finally, for those $z'$ associated with $x'\in\left\{\frac{n-N}{2},\dots,\frac{n+N}{2}\right\}$ we have from Eq.~\eqref{eq:P stoch exact} that:
\begin{equation}
\sum_{i=1}^{d}\sum_{x=0}^{n}\sum_{z\in s_x} R\left(i,z|j,z'\right)=1,
\end{equation}
i.e. these columns do actually sum to 1.

\subsubsection{Construction of a doubly stochastic matrix}
\label{sec:doublystoch}
Finally we wish to construct a doubly stochastic matrix from $R$ which also maps the Schmidt coefficients of $\ket{\Phi_N}$ to those of $\ket{\Psi_N}$. This will directly imply the existence of the LOCC protocol, $\Lambda_N$, taking $\ket{\Psi_N}$ to $\ket{\Phi_N}$. We will denote this matrix by $\tilde{R}$ and construct it as follows:
\begin{enumerate}
\item For all $z'$ associated with $x'\in\left\{\frac{n-N}{2},\dots,\frac{n+N}{2}\right\}$, set:
\begin{equation}
\tilde{R}\left(i,z|j,z'\right):=R\left(i,z|j,z'\right).
\end{equation}
\item Define:
\begin{equation}
r_{i,z}=\sum_{j=1}^{d} \sum_{x'=\frac{n-N}{2}}^{\frac{n+N}{2}} \sum_{z'\in s_{x'}} R\left(i,z|j,z'\right).
\end{equation}
Then for $z'$ associated with $x'\notin\left\{\frac{n-N}{2},\dots,\frac{n+N}{2}\right\}$, set:
\begin{equation}
\tilde{R}\left(i,z|j,z'\right):=\frac{1-r_{i,z}}{d M},
\end{equation}
where $M:=\sum_{x'=0}^{\frac{n-N}{2}}u^{x'}\left(u-1\right)^{n-x'}+\sum_{x'=\frac{n+N}{2}}^{n}u^{x'}\left(u-1\right)^{n-x'}$ so $d M$ is the number of columns of $\tilde{R}$ not contained in the support of the Schmidt coefficients of $\ket{\Phi_N}$. In other words, this procedure evenly distributes the deficit in each row amongst the columns of $\tilde{R}$ not contained in the support of the Schmidt coefficients of $\ket{\Phi_N}$.
\end{enumerate}
Now, for $z'$ associated with $x'\in\left\{\frac{n-N}{2},\dots,\frac{n+N}{2}\right\}$, we have:
\begin{equation}
\sum_{i=1}^{d}\sum_{x=0}^{n}\sum_{z\in s_x}\tilde{R}\left(i,z|j,z'\right)=1,
\end{equation}
while for all other $z'$:
\begin{align*}
\sum_{i=1}^{d}\sum_{x=0}^{n}\sum_{z\in s_x}\tilde{R}\left(i,z|j,z'\right)&=\sum_{i=1}^{d}\sum_{x=0}^{n}\sum_{z\in s_x}\frac{1-r_{i,z}}{d M}\\
&=\frac{d M_T}{d M}-\frac{1}{d M}\left(d M_T-d M\right)\\
&=1,
\end{align*}
where $M_T:=\sum_{x=0}^{n}u^x\left(u-1\right)^{n-x}$ so $d M_T$ is the total number of columns/rows in $\tilde{R}$. Hence $\tilde{R}$ is stochastic.

By construction, we have that:
\begin{equation}
\sum_{j=1}^{d}\sum_{x'=0}^{n}\sum_{z'\in s_{x'}} \tilde{R}\left(i,z|j,z'\right)=1.
\end{equation}
Hence $\tilde{R}$ is doubly stochastic.

Finally, as we have not altered the columns in the support of $\Phi_N$, we have that:
\begin{align*}
&\sum_{j=1}^{d}\sum_{x'=\frac{n-N}{2}-a^{u}_{\textrm{max}}}^{\frac{n+N}{2}+a^{u}_{\textrm{max}}}\sum_{z\in s_{x'}} \tilde{R}\left(i,z|j,z'\right) \frac{q_j}{N+1} u^{-x'}\left(u-1\right)^{x'-n} \delta_{x'\in\left\{\frac{n-N}{2},\dots,\frac{n+N}{2}\right\}}\\
=&\sum_w \frac{p_{iw}}{N+1} u^{-x}\left(u-1\right)^{x-n} \delta_{x+a^{u}_{w}\in\left\{\frac{n-N}{2},\dots,\frac{n+N}{2}\right\}},
\end{align*}
so $\tilde{R}$ maps the Schmidt coefficients of $\ket{\Phi_N}$ to those of $\ket{\Psi_N}$.
Using the results of \cite{Nielsen-pure-entanglement}, the existence of such a $\tilde{R}$ implies that we have an LOCC protocol that converts $\ket{\Psi_N}$ into $\ket{\Phi_N}$.

\section{The large \texorpdfstring{$N$}{N} limit of BLOCC protocols. }
\label{ssec:large n}

Here we show that the state $\ket{\Psi_N}$ given in Eq.~\eqref{eq:initialN}:
\begin{equation}
\ket{\Psi_N}=\sum_{i=1}^{d}\sum_{x=0}^{n} \sqrt{\sum_w \frac{p_{iw}}{N+1} \delta_{x+a^{u}_{w}\in\left\{\frac{n-N}{2},\dots,\frac{n+N}{2}\right\}}} \ket{i,i}\otimes \ket{e_x},
\end{equation}
 tends to a state that is product between system and battery in the limit of large $N$. To see this, consider the overlap between $\ket{\Psi_N}$ and the state:
\begin{equation}\label{eq:tpsin}
\ket{\tilde{\Psi}_N}=\sum_{i=1}^{d}\sum_{x=0}^{n} \sqrt{\frac{p_i}{n+1}}\ket{ii}\otimes \ket{e_x},
\end{equation}
It is given by:
\begin{align}
\braket{\tilde{\Psi}_N}{\Psi_N}=&\sum_{i=1}^{d}\sum_{x=0}^{n}\sqrt{\frac{p_i}{n+1}\sum_w \frac{p_{iw}}{N+1}\delta_{x+a^{u}_{w}\in\left\{\frac{n-N}{2},\dots,\frac{n+N}{2}\right\}}}\nonumber \\
\geq& \sum_{i=1}^{d}\sum_{x=\frac{n-N}{2}+a^{u}_{\textrm{max}}}^{\frac{n+N}{2}-a^{u}_{\textrm{max}}} \sqrt{\frac{p_i^2}{\left(n+1\right)\left(N+1\right)}}\nonumber \\
\geq& \frac{n+1-2a^{u}_{\textrm{max}}}{n+1}\nonumber \\
=&\frac{N+1}{N+1+2a^{u}_{\textrm{max}}}\label{eq:finitebound} \\ \nonumber
\stackrel{N\rightarrow\infty}{\longrightarrow}&1.
\end{align}
Hence, the fidelity between $\ket{\tilde{\Psi}_N}$ and $\ket{\Psi_N}$ tends to 1. Thus, in the limit of large $N$, the initial state of the system tends to the pure state $\ket{\psi}=\sum_{i=1}^{d}\sqrt{p_i}\ket{ii}$, the reduced state of the system in $\ket{\tilde{\Psi}_N}$.

Note that given a protocol $\Lambda_N$ that converts $\ket{\Psi_N}$ to $\ket{\Phi_N}$, if we apply $\Lambda_N$ to $\ket{\tilde{\Psi}_N}$ we will in general create a mixed state $\tilde{\sigma}_N$. However, as the fidelity is non-decreasing under the application of quantum channels, the fidelity between $\ket{\Phi_N}$ and $\tilde{\sigma}_N$ will also tend to 1 with increasing $N$ and in addition the reduced state of $\tilde{\sigma}_N$ on the system $AB$ will be increasingly close in fidelity to $\ket{\phi}_{AB}$.

We can also consider the closeness of the probability distribution:
\begin{equation}
\tilde{P}\left(i,x,j,x'\right)=\tr\left[\ketbra{j}{j}\otimes P_{x'} \sum_m M_m \left(\ketbra{i}{i}\otimes P_x\right) \tilde{\Psi}\left(\ketbra{i}{i}\otimes P_x\right) M_m^\dagger \right]
\end{equation}
to
\begin{equation}
{P}\left(i,x,j,x'\right)=\tr\left[\ketbra{j}{j}\otimes P_{x'} \sum_m M_m \left(\ketbra{i}{i}\otimes P_x\right) {\Psi}\left(\ketbra{i}{i}\otimes P_x\right) M_m^\dagger \right].
\end{equation}
If the trace distance between $\tilde{\Psi}$ and $\Psi$, $D\left(\Psi,\tilde{\Psi}\right)$, is $\epsilon$, then as for general $\rho$ and $\sigma$ we have $D\left(\rho,\sigma\right)=\max_{\left\{E_m\right\}} D\left(s_m,t_m\right)$ (where the maximisation is over all POVMS and where $s_m:=\tr\left[M_m \rho\right]$ and $t_m:=\tr\left[M_m \sigma\right]$), we have that:
\begin{equation}
\sum_{i,j,x,x',m} \left|\tilde{P}\left(i,x,j,x',m\right)-P\left(i,x,j,x',m\right)\right|\leq \epsilon
\end{equation}
where
\begin{equation}
{P}\left(i,x,j,x',m\right)=\tr\left[\ketbra{j}{j}\otimes P_{x'} M_m \left(\ketbra{i}{i}\otimes P_x\right) {\Psi}\left(\ketbra{i}{i}\otimes P_x\right) M_m^\dagger \right]
\end{equation}
and $\tilde{P}\left(i,x,j,x',m\right)$ is similarly defined. This implies that
\begin{equation}
\sum_{x} \left|\tilde{P}\left(i,x,j,x'\right)-P\left(i,x,j,x'\right)\right|\leq \epsilon
\end{equation}
and finally that
\begin{equation}
\left|\tilde{P}\left(i,j,w\right)-P\left(i,j,w\right)\right|\leq \epsilon
\end{equation}
so we obtain a similar distribution when applying the LOCC protocol to $\ket{\tilde{\Psi}}$ as if we had applied it to $\ket{{\Psi}}$.

\section{The 2nd law equality for entanglement} \label{sec:2nd eq}

In this section we give the proof of Result~\ref{res:2nd eq}. 

\begin{theorem}
Given an initial state $\ket{\psi}$ with Schmidt coefficients $p_i$ and a target state $\ket{\phi}$ with coefficients $q_j$, the distribution of entanglement that can be coherently extracted in converting $\ket{\psi}$ into $\ket{\phi}$ under BLOCC satisfies:
\begin{equation}
\left\langle 2^{w-\log q_j+\log p_i}\right\rangle=1.
\end{equation}
\end{theorem}
\begin{proof}
The proof follows straightforwardly from the constraints on stochastic matrices from Result \ref{result:bat-stochastic}.
As we are considering BLOCC protocols, Eq.~\eqref{eq:w-stochasticity} holds:
\begin{equation*}
\sum_{j,w} P\left(i,w|j\right)2^w=1,\quad \forall i.
\end{equation*}

Multiplying this equation by $p_i$ and summing over $i$ then gives (with a small rewriting of the conditional probability distribution):
\begin{equation*}
\sum_{i,j,w} P\left(i,j,w\right) \frac{p_i}{q_j}2^w=1.
\end{equation*}
Moving the probabilities into the exponent then gives the result.
\end{proof}

\section{A quantitative third law of entanglement}
\label{sec:third}

Here we give the detailed proof of Result \ref{re:3rdlaw}.
From the majorization criterion, we know that in pure to pure transitions, the Schmidt rank cannot increase, not even probabilistically \cite{vidal1999entanglement}. This is essentially the analogue of a number of results in thermodynamics associated with the 3rd law, where a general statement is that decreasing the rank of a state requires infinite resources, in the form of infinite work fluctuations, an infinite-sized bath or both \cite{masanes2017general} . The particular question that this answers is: what is the infinite resource involved in a potential increase of Schmidt rank? 

\begin{theorem}
	Let $p_{\text{min}}$ and $q_{\text{min}}$ be the smallest Schmidt coefficients of the initial and final states of the system. The entanglement fluctuations are bounded by
	\begin{equation}
	\sum_w 2^w \ge \frac{q_{\text{min}}}{d' p_{\text{min}}},
	\end{equation}
	where $d'$ is the number of nonzero Schmidt coefficients of the final state.
\end{theorem}

\begin{proof}
	We start from Eq. \eqref{eq:doesthetrick}
	\begin{equation}
	\sum_{j,w} P\left(i_0,w|j\right)q_j=p_{\text{min}},
	\end{equation}
	with $p_{i_0}=p_{\text{min}}$. From this we can write
	\begin{equation}
	P\left(i_0,w|j\right) q_j \le p_{\text{min}} \,\,\,\, \forall j.
	\end{equation}
	Plugging this in Eq. \eqref{eq:w-stochasticity} we obtain
	\begin{equation}
	\sum_{j,w} \frac{p_{\text{min}}}{q_j}2^w \ge 1,
	\end{equation}
	from which it follows that 
	\begin{equation}
	d'\frac{p_{\text{min}}}{q_{\text{min}}}\sum_w 2^w \ge 1,
	\end{equation}
	finishing the proof.
\end{proof} 
From this, we see that if we start in a state with a small lowest probability and transform it into one in which the probabilities are more uniform, either the magnitude of the biggest entanglement fluctuations will have to be very large, or there will be a large number of fluctuations. In particular, we see that in the limit in which we are increasing the Schmidt rank (that is, when $p_{\text{min}}\rightarrow 0$), the amount of entanglement which might be required, must diverge.

It is important to notice that our framework, the accuracy of the approximations is limited by the magnitude of the biggest work fluctuation $w_{\text{max}}$, as in Eq. \eqref{eq:finitebound}. On top of that, the number of possible work fluctuations is limited by the size of the battery we are using. These two factors hence limit how much can we change a very small Schmidt coefficient.

\section{Jarzynski's equality for entanglement} \label{sec:jar}

In this section we give the proof of Result~\ref{res:jar}.

\begin{theorem}
When the final state is a maximally entangled states of dimension $d'$, we have:
\begin{equation}
\left\langle 2^w \right\rangle = \frac{d}{d'},
\end{equation}
with $d$ the dimension of the support of the initial state.
\end{theorem}
\begin{proof}
We again start from Eq.~\eqref{eq:w-stochasticity}:
\begin{equation*}
\sum_{j,w} P\left(i,w|j\right)2^w=1,\quad \forall i.
\end{equation*}
We have that $P\left(i,w|j\right)\frac{1}{d}=P\left(i,w,j\right)$. Hence if we multiply both sides with $\frac{1}{d'}$ we obtain
\begin{equation*}
\sum_{j,w} P\left(i,w,j\right)2^w=\frac{1}{d'},\quad \forall i.
\end{equation*}
We now sum over the index $i$, to obtain
\begin{equation*}
\sum_{i,j,w} P\left(i,w,j\right)2^w=\sum_i \frac{1}{d'}=\frac{d}{d'}.
\end{equation*}
\end{proof}

\section{Reversed transformations and an entanglement Crooks theorem}
\label{sec:crooks}

\subsection{Crooks' fluctuation theorem}\label{sec:appcrooks}

 Here we outline an important result in statistical mechanics for which we are giving an entanglement analogue. This is Crooks' theorem, first shown for classical settings in the seminal paper \cite{crooks1999entropy}, and later extended to quantum systems by Tasaki \cite{tasaki2000jarzynski}. 

The setting is as follows: a system is in an initial thermal state $\frac{e^{-\beta H}}{Z}$, with $Z=\tr\left[e^{-\beta H}\right]$. It is then taken out of equilibrium through some particular protocol (for instance, an unitary process). An amount of work $W$ is consumed in the process, and this quantity can vary within different runs, giving rise to a probability distribution $P(W)$. At the end of the protocol, the Hamiltonian of the system may have changed to $H'$.

Then the system is reset to the new thermal equilibrium $\frac{e^{-\beta H'}}{Z'}$, and a \emph{time-reversed} protocol is applied to it, extracting a work distribution $P^{\text{rev}}(-W)$.

Crooks' theorem then relates the two work distributions via the following relation
\begin{equation}\label{eq:crooksthermo}
\frac{P(W)}{P^{\text{rev}}(-W)}=e^{-\beta W} \frac{Z'}{Z} .
\end{equation}
 This is thus a relation between the work extraction of two different processes starting from thermal equilibrium. It expresses the fact that extracting positive work along a process has a probability which is exponentially suppressed with respect to that of extracting a negative amount of work in the reversed process. This way, it can be understood as a quantitative statement of the irreversibility of thermodynamics.



\subsection{Reversed LOCC}\label{sec:reversed}

We now proceed to define the analogue of the reversed LOCC protocol from which we will derive an entanglement version of Crooks' theorem. The idea is to define a protocol in which the unitaries performed by Bob are not given by $V_m$ but by their conjugates $V^\dagger_m$. This then yields a simple relation between the stochastic matrices that correspond to each of the two processes..

Let us start with Eq.~\eqref{eq:forwardp1} from Section \ref{sec:necc}:
\begin{equation}\label{eq:forwardp}
(M_m \otimes V_m) 
\ket{\Psi}_{ABA'B'} 
=
\sqrt {r_m}\,  
\ket{\Phi}_{ABA'B'}\ .
\end{equation}
We can rewrite this equation as
\begin{equation}
\left(\left(M_m \sqrt{\Psi_{AA'}}\right) \otimes V_m\right) 
\ket{\xi}_{ABA'B'} 
=
\sqrt {r_m}\left(\sqrt{\Phi_{AA'}}\otimes \id \right) 
\ket{\xi'}_{ABA'B'}\,
\end{equation}
where $\ket{\xi}_{ABA'B'}=\sum_l \ket{ll}_{ABA'B'}$ and $\ket{\xi'}_{ABA'B'}=\sum_l \ket{l'l'}_{ABA'B'}$ are the un-normalized maximally entangled states in the Schmidt basis of the initial and final states. 

 Thus we also have (using the Choi-Jamio\l{}kowski isomorphism)
\begin{equation}
M_m \sqrt{\Psi_{AA'}} \sum_l \ketbra{l}{l}_{AA'} V^\dagger_m
=
\sqrt {r_m}\,\sqrt{\Phi_{AA'}} \sum_{l'} \ketbra{l'}{l'}_{AA'},
\end{equation}
which is equivalent to the following operator identity in the Hilbert space of $AA'$ 
\begin{equation}\label{eq:polardec}
M_m \sqrt{\Psi_{AA'}}
=
\sqrt {r_m}\,\sqrt{\Phi_{AA'}} V_m .
\end{equation}
This gives the polar decomposition of the operator $M_m \sqrt{\Psi_{AA'}}$.

Let us now choose an arbitrary state $\Phi'_{AA'}$ which commutes with $\Psi_{AA'}$ (such that they have the same Schmidt basis), and define $\Psi'_{AA'}$ as
\begin{equation}\label{eq:revmixture}
\Psi'_{AA'}=\sum_m r_m V_m \Phi'_{AA'} V^\dagger_m.
\end{equation}
Note that $\Psi'_{AA'}$ has the same eigenbasis as $\Phi_{AA'}$ while $\Phi'_{AA'}$ with the same eigenbasis as $\Psi_{AA'}$. This follows from the fact that the unitaries $V_m$ map between the two bases, as can be seen from Eqs. \eqref{eq:forwardp} and \eqref{eq:polardec}.

These new states allow us to define the following set of positive operators
\begin{equation}
M^{\text{rev}}_m := \sqrt{r_m} \sqrt{\Phi'_{AA'}} V^\dagger_m  \sqrt{\Psi'_{AA'}}^{-1},
\end{equation}
where $\sqrt{\Psi'_{AA'}}^{-1}$ has non-zero support on the support of $\Psi'_{AA'}$ only, the projector onto which we define as $\Pi_{\Psi'_{AA'}}$. 

It is straightforward to check that
\begin{equation}
\sum_m M^{\text{rev} \dagger}_m  M^{\text{rev}}_m =\Pi_{\Psi'_{AA'}}.
\end{equation}
Thus, together with the projector $\id-\Pi_{\Psi'_{AA'}}$, they form a valid POVM. 
They also satisfy the following identity:
\begin{equation}
M^{\text{rev} \dagger}_m \Psi'_{AA'} M^{\text{rev}}_m = \sqrt{r_m} \Phi'_{AA'},\quad\forall m.
\end{equation}
Note that the outcome given by $\id-\Pi_{\Psi'_{AA'}}$ never occurs when acting on $\Psi'_{AA'}$.

We now show that these measurement operators give a pure state when applied to $\ket{\Psi}_{AA'BB'}$. Note that
\begin{align}
	\ket{\Phi'_m}&=\frac{1}{\sqrt{r_m}}\left(M_m^\text{rev}\otimes \id\right) \ket{\Psi'}_{AA'BB'} \\ 
	&=\left(\sqrt{ {\Phi'}_{AA'}} \otimes \id\right) \left( V_m^{\dagger}\otimes \id\right) \sum_{l'}\ket{l'l'}_{ABA'B'}
	\label{eq:mtrans}
\end{align}
so we see that to obtain $\ket{\Phi'}_{AA'BB'}$ Bob has to implement the particular unitary that maps the initial to the final Schmidt basis. This unitary is  $V^\dagger_m$, as seen in Eq. \eqref{eq:mtrans}. Hence we have
\begin{align}
\id \otimes V^\dagger_m	\ket{\Phi'_m} & =\sqrt{ {\Psi'}_{AA'}} \otimes \id \sum_{l'} V_m^{\dagger}\otimes V_m^\dagger \ket{l'l'}_{ABA'B'}\\
& = \sqrt{ {\Phi'}_{AA'}} \otimes \id \sum_{l} \ket{ll}_{ABA'B'}= \ket{\Psi'}_{AA'BB'}.
\end{align}
Note that the second line follows from the first because of the definition of $V^\dagger_m$ in Eq.\eqref{eq:forwardp} we know that $V_m^{\dagger}\otimes V_m^\dagger$ is the unitary that maps the initial to the final Schmidt basis in both Alice and Bob simultaneously (up to a permutation of the elements which may depend on $m$).

We have thus shown that
	\begin{equation}\label{eq:backprotocol}
	M^{\text{rev}}_m \otimes V^\dagger_m \ket{\Psi'}_{ABA'B'}=\sqrt{r_m} \ket{\Phi'}_{ABA'B'},\quad \forall m
	\end{equation}
That is, we have defined a reversed LOCC protocol that takes state $\ket{\Psi'}_{AA'BB'}$ to $\ket{\Phi'}_{AA'BB'}$. In the next section we move onto investigating the relationship between the original LOCC transformation and this reversed protocol.

\subsection{Crooks' theorem for entanglement}\label{sec:crooks2}

In this section we show how the notion of the reversed protocol allows us to derive an entanglement analogue of Crooks theorem. For simplicity, we shall assume that for the work distributions extracted, $\frac{w}{\delta w}$ is an integer.

Let us assume that we have a sequence of forward protocols that takes $\ket{\Psi_N} \rightarrow \ket{\Phi_N}$ as defined in Eq.~\eqref{eq:initialN} and Eq.~\eqref{eq:finalN}, in which there is a work distribution. The results in Section \ref{sec:necc} imply that we can find a matrix $P\left(i,w|j\right)$. After that, we can use the results of Section \ref{sec:constructionlambda} to define $P\left(i,x|j,x'\right)$, $R\left(i,z|j,z'\right)$ and $\tilde R\left(i,z | j,z'\right)$. The matrix $\tilde R\left(i,z | j,z'\right)$ is the stochastic matrix that changes the Schmidt coefficients of the final state to those of the initial state, and thus we can write it as
\begin{equation}
\tilde R\left(i,z | j,z'\right)=\tr\left[\left(\proj{i}\otimes \proj{z}\right) \Lambda^*_{BB'} \left(\proj{j}\otimes \proj{z'}\right)\right],
\end{equation}
where $\Lambda^*_{BB'}\left(\cdot\right)$ is defined as in Eq. \eqref{eq central}. 

Let us now define the following matrix
\begin{equation}
Q\left(i,x | j,x'\right)=\tr\left[\left(\proj{i}\otimes P_x\right) \Lambda^*_{BB'} \left(\proj{j}\otimes s_{x'}\right)\right].
\end{equation}
For the range of $x'$ in which the battery has support, that is $x' \in \left\{\frac{n-N}{2}...\frac{n+N}{2} \right\}$, this is related to $P\left(i,x | j,x'\right)$ as
\begin{align}
Q\left(i,x | j,x'\right)&=\sum_{z \in s_x,z' \in s_{x'}} u^{-x'}\left(u-1\right)^{x'-n} \tr \left[\left(\proj{i}\otimes \proj{z}\right) \Lambda ^*_{BB'} \left(\proj{j}\otimes \proj{z'}\right)\right] \\ &=
\sum_{z \in s_x,z' \in s_{x'}} \tilde R \left(i,z | j,z'\right) \\ &=
\sum_{z \in s_x,z' \in s_{x'}} u^{-x'}\left(u-1\right)^{x'-n} P(i,x | j,x') \\ &=
P\left(i,x |j,x'\right).
\end{align}
The step from the second to the third line only holds in that particular range of $z'$ (or rather, $x'$) specified above, in which $\tilde R(i,z|j,z') =R(i,z|j,z')$. 

Thus, using Eq. \eqref{eq:pix}, we have that
\begin{equation}\label{eq:qpw}
Q(i,x | j,x')=\sum_w P(i,w |j) \delta_{x'-x,\frac{w}{\delta w}}.
\end{equation}
within this range of $x'$.

In the previous section we have seen that the reversed LOCC protocol is such that the mixture of unitaries is the dual. This motivates the following definition
\begin{align}
Q^{\text{rev}}\left(j,x' | i,x\right)&=\tr\left[\left(\proj{j}\otimes P_{x'}\right) \Lambda_{BB'} \left(\proj{i} \otimes s_x\right)\right].
\end{align}
This matrix satisfies:
\begin{align}\label{eq:precrooks}
Q^{\text{rev}}\left(j,x' | i,x\right)&=\frac{u^{x'} \left(u-1\right)^{n-x'}}{u^{x} \left(u-1\right)^{n-x}} Q\left(i,x|j,x'\right)
\end{align}

Using Eq. \eqref{eq:qpw}, if we assume that $x'$ is within the range in which the battery has support, that is $x' \in \left\{\frac{n-N}{2}....\frac{n+N}{2} \right\}$, we have that
\begin{align}\label{eq:qtoprev}
Q^{\text{rev}}(j,x' | i,x)&=\frac{u^{x'} \left(u-1\right)^{n-x'}}{u^{x} \left(u-1\right)^{n-x}} \sum_w P\left(i,w |j\right) \delta_{x'-x,\frac{w}{\delta w}} \\&=
\frac{u^{x'} \left(u-1\right)^{n-x'}}{u^{x} \left(u-1\right)^{n-x}} \sum_w  \delta_{x'-x,\frac{w}{\delta w}} \sum_{x''} \alpha_{x''}\tr\left[\left(\proj{i}\otimes P_{x''-\frac{w}{\delta w}}\right)\Lambda^*_{BB'}\left(\proj{j}\otimes s_{x''}\right)\right]
\\&=
\sum_w  \delta_{x'-x,\frac{w}{\delta w}} \sum_{x''} \left(\frac{u}{u-1}\right)^{\frac{w}{\delta w}}  \alpha_{x''}\tr\left[\left(\proj{i}\otimes P_{x''-\frac{w}{\delta w}}\right)\Lambda^*_{BB'}\left(\proj{j}\otimes s_{x''}\right)\right]
\\&=
\sum_w  \delta_{x'-x,\frac{w}{\delta w}} \sum_{x''}
\alpha_{x''}\tr\left[\left(\proj{j}\otimes P_{x''}\right)\Lambda_{BB'}\left(\proj{i}\otimes s_{x''-\frac{w}{\delta w}}\right)\right]
\\& \simeq
\sum_w  \delta_{x'-x,\frac{w}{\delta w}} \sum_{x''}
\alpha_{x''+\frac{w}{\delta w}}\tr\left[\left(\proj{j}\otimes P_{x''+\frac{w}{\delta w}}\right)\Lambda_{BB'}\left(\proj{i}\otimes s_{x''}\right)\right]
\\& \simeq
\sum_w  \delta_{x'-x,\frac{w}{\delta w}} \sum_{x''}
\alpha_{x''}\tr\left[\left(\proj{j}\otimes P_{x''+\frac{w}{\delta w}}\right)\Lambda_{BB'}\left(\proj{i}\otimes s_{x''}\right)\right]
\\&
\equiv \sum_w P^{\text{rev}}\left(j, -w |i\right) \delta_{x'-x,\frac{w}{\delta w}}, 
\end{align}
where the approximations are exact in the limit of an ideal battery, and in the last line we have defined $P^{\text{rev}}\left(j, -w |i\right)\equiv \sum_{x''}
\alpha_{x''}\tr\left[\left(\proj{j}\otimes P_{x''+\frac{w}{\delta w}}\right)\Lambda_{BB'}\left(\proj{i}\otimes s_{x''}\right)\right]$ in analogy to Eq.~\eqref{eq:relationP2}. This is such that
\begin{align}
\sum_{j,w} P^{\text{rev}}\left(j, -w |i\right)=1 \\
\sum_{i,w} P^{\text{rev}}\left(j, -w |i\right) 2^{-w}=1,
\end{align}
which follow from Eq.~\eqref{eq:normalisation} and Eq.~\eqref{eq:w-stochasticity}  respectively.
Because it is a stochastic matrix, it maps an arbitrary probability distribution to another
\begin{equation}
\sum_{i,w} P^{\text{rev}}\left(j, -w |i\right) q'_i=p'_{j,-w}. 
\end{equation}

To summarize:
given a sequence of forward LOCC protocols, with a matrix $Q(i,x|j,x')$ that maps the final state coefficients to the initial state ones, there exists a sequence of reversed protocols as defined in Section \ref{sec:reversed}, with a matrix given by $Q^{\text{rev}}(j, x' |i,x)$ that also maps from the final coefficients to the initial ones.
Let us now take a final state for the protocol to be
\begin{equation} \label{eq:Phi'}
\ket{\Phi'_N}=\sum_{i=1}^{d} \sum_{x=0}^n \sqrt{\frac{q'_i}{N'+1}}\delta_{x \in \left\{\frac{n-N'}{2},....,\frac{n+N'}{2}\right\}} \ket{ii} \otimes \ket{e_x},
\end{equation}
where $N'=N-2\left| \frac{w_{\text{max}}}{\delta w} \right|$, $\left| \frac{w_{\text{max}}}{\delta w} \right|$ is the absolute value of the integer corresponding to the biggest work fluctuation and $N$ is related to $n$ as per Section~\ref{sec:suff}. The reasoning behind choosing the battery support to be in terms of $N'$ rather than $N$ will be explained shortly.

These coefficients are mapped to the following initial state coefficients
\begin{align}
\sum_{i,x} Q^{\text{rev}}\left(j,x' | i,x\right) \frac{q'_i}{N'+1}\delta_{x \in \left\{\frac{n-N'}{2},....,\frac{n+N'}{2}\right\}} &= \sum_{i,x} \sum_w
P^{\text{rev}}\left(j,-w|i\right)\delta_{x-x',\frac{w}{\delta w}} \frac{q'_i}{N'+1}\delta_{x \in \left\{\frac{n-N'}{2},....,\frac{n+N'}{2}\right\}} \\&=
\sum_i \sum_w 
P^{\text{rev}}\left(j,-w|i\right) \frac{q'_i}{N'+1}\delta_{x'-\frac{w}{\delta w} \in \left\{\frac{n-N'}{2},....,\frac{n+N'}{2}\right\}} \\&=
\sum_w \frac{p'_{j,-w}}{N'+1}\delta_{x'-\frac{w}{\delta w}\in \left\{\frac{n-N'}{2},....,\frac{n+N'}{2}\right\}}
\end{align}
Thus in analogy with the results of Section \ref{sec:suff} we conclude that the sequence of reversed protocols maps from the following initial states
\begin{equation}\label{eq:initialstaterev}
\ket{\Psi'_N}=\sum_{j=1}^{d} \sum_{x'=0}^n \sqrt{\sum_w \frac{p'_{j,-w}}{N'+1}\delta_{x'-\frac{w}{\delta w}\in \{\frac{n-N'}{2},....,\frac{n+N'}{2}\}}}\ket{jj}\otimes \ket{e_{x'}},
\end{equation}
to the $\ket{\Phi'_N}$ given in Eq.~\eqref{eq:Phi'}. 
The correction in the support of the battery to the range $\left\{\frac{n-N'}{2},....,\frac{n+N'}{2}\right\}$ is such that in the sum Eq. \eqref{eq:initialstaterev} the variable $x'$ does not take values outside the range  $\left\{\frac{n-N}{2},....,\frac{n+N}{2}\right\}$, which is the condition needed for Eq. \eqref{eq:qtoprev} to hold.

We are now in a position to derive the analogue of Crooks' theorem. While in thermodynamics the derivation of Crooks requires that the initial states of both the forward and reverse protocols are thermal, for entanglement we need to take the final states of both protocols to be maximally entangled (though possibly of different dimensions) so $q_j=\frac{1}{d}$ and $q'_i=\frac{1}{d'}$.

The work distributions associated with the forward and reversed processes are then:
\begin{align}
P\left(w\right)&=\sum_{i,j} P\left(i,w |j\right) \frac{1}{d} \\
P^{\text{rev}}\left(-w\right)&=\sum_{i,j} P^{\text{rev}}\left(j,-w |i\right) \frac{1}{d'}.
\end{align}
Following from Eq.~\eqref{eq:precrooks}, it can be seen that they obey the relation 
\begin{equation}
\frac{P\left(w\right)}{P^{\text{rev}}\left(-w\right)}=2^{-w} \frac{d'}{d}.
\end{equation}

We have thus shown:
\begin{theorem}
Suppose $\ket{\psi}\stackrel{\textrm{BLOCC}}{\longrightarrow}\ket{\textrm{ebit}_d}$ while extracting entanglement $\left\{P\left(w\right),w\right\}$. Then there exists a state $\ket{\psi'}$ such that $\ket{\psi'}\stackrel{\textrm{BLOCC}^{\textrm{rev}}}{\longrightarrow}\ket{\textrm{ebit}_{d'}}$ while extracting entanglement $\left\{P^{\textrm{rev}}\left(-w\right),w\right\}$ and where:
\begin{equation}
\frac{P\left(w\right)}{P^{\text{rev}}\left(-w\right)}=2^{-w} \frac{d'}{d}.
\end{equation}
\end{theorem}

\section{Probabilistic BLOCC transformations and full reversibility of entanglement dilution and concentration}
\label{sec:ensembles}

We previously considered pure state transformations given by Equation \eqref{eq:puretransbat} where the final state is close to $\ket{\phi}_{AB}\otimes \ket{\bat'}_{A'B'}$. We can generalise our results to the case when the final state is given by an ensemble of pure states close to $\ket{\phi^t}_{AB}\otimes \ket{\bat^t}_{A'B'}$ each with probability $p_t$. Or, to put it in a way which holds in the idea limit, where the amount of entanglement transferred to the battery is a classical random variable $p_{w|t}$ occurring with probability $p_t$ and the final system states are an ensemble
$\ket{\phi^t}_{AB}$.

Let us call transformations with such an ensemble of pure states, Probabilistic BLOCC.
We want to show that Equations \eqref{eq:normalisation}-\eqref{eq:doesthetrick} are still necessary and sufficient, even if the final battery state is allowed to be an ensemble of pure states. This allows us to apply our results to well known examples such as entanglement concentration, where one does not have an entanglement battery, and instead one distills some random number of ebits peaked around $nS(\psi_A)$ from $n$ copies of $\ket{\psi}_{AB}$ \cite{BBPS1996}. Then the case of entanglement concentration corresponds to setting all the $\ket{\bat^t}_{A'B'}$  equal to the initial state $\ket{\bat}_{A'B'}$ (so the battery is not used in the transformation), and by considering the final system to be some number of maximally entangled states $\ket{\phi^t}_{AB}=\ket{e_t}_{AB}$. One can then transfer these ebits into the battery so that the final state of the system is in a product state, and the battery is in $\Delta^t\ket{\bat}_{A'B'}$ with probability $p_t$.

To see that \eqref{eq:normalisation}-\eqref{eq:doesthetrick} are still necessary and sufficient conditions, we can use a result which follows from \cite{JonathanP}:
\begin{lemma} Consider an ensemble of pure states $\ket{\phi_t}$ occurring with probability $p_t$ which can be written in a Schmidt basis as $\ket{\phi_t}=\sum_j \sqrt{q_{j|t}}\ket{jj}_{AB}$. Consider the {\it average target state}, $\ket{\bar{\phi}}=\sum_j \sqrt{q_{j}}\ket{jj}_{AB}$ where $q_j=\sum_t p_t q_{j|t}$. Then it is possible to transform an initial state $\ket{\psi}_{AB}$ to the ensemble  $\{ \ket{\phi_t},p_t\}$ under LOCC iff it's possible to transform  $\ket{\psi}_{AB}$ into $\ket{\bar{\phi}}$.
\label{lemma-JP}
\end{lemma}
Note that because all states with the same Schmidt coefficients are equivalent under LOCC, we can write $\ket{\phi_t}$ using the same Schmidt basis without loss of generality.

Now, to see that Equations \eqref{eq:normalisation}-\eqref{eq:doesthetrick} are necessary under Probabilistic BLOCC, consider an ensemble of possible work and final state distributions, which we can characterize by the probability distributions $\{ P(w,j|i,t),p_t\}$, with 
the ensemble of final states arbitrarily close to a pure state ensemble which we denote by $\{p_t,\ket{\phi^t}\otimes\ket{\bat^t}\}$, each with Schmidt coefficients $\sqrt{q_{jx'|t}}\, [(u-1)/u]^{x'/2}$. 
(In the rest of this section we omit the factor $[(u-1)/u]^{x'/2}$.)
We can write the process map which takes $p_{xi}$ to $q_{x'j|t}$ as $P(j,x'|i,x,t) = P(w= [x'-x]\delta w ,j|i,t)$.
Then, due to Lemma \ref{lemma-JP}, we know that one can transform the initial state into this ensemble, only if one can transform it into the  average target state with Schmidt coefficients  $\sqrt{\sum_tq_{jx'|t}p_t}$.
Since we can take $\sum_t q_{jx'|t}p_t=q_{jx'}$, the transformation into the ensemble can only be accomplished if we can transform into the pure state with Schmidt coefficients $\sqrt{q_{jx'}}$ of the average state. The necessary conditions for pure state transformations, Equations \eqref{eq:normalisation}-\eqref{eq:doesthetrick} then apply.

To see that  \eqref{eq:normalisation}-\eqref{eq:doesthetrick} are sufficient, we want that given any $P(w,j|i)$ that satisfies them, one can transform the initial state into any
ensemble given by the process map $P(w,j|i,t)$ as long as $P(w,j|i)=\sum_t P(w,j|i,t)p_t$. Now, every
 $P(w,j|i,t)$ can be considered as a process which takes the initial state with
 Schmidt coefficients $p_{ix}$ to a final state with Schmidt coefficients $q_{jx'|t}$, which we can convert to a $P(x',j|i,x,t)$ as was done in Equation \eqref{eq:pix}. We can then see that $P(x',j|i,x)=\sum_t P(x',j|i,x,t)p_t$, and that
 $q_{jx'}=\sum_t q_{jx'|t}p_t$ gives the Schmidt coefficients for the average state
 corresponding to the final ensemble given by $q_{jx'|t}$. Lemma \ref{lemma-JP} then guarantees that we can create this ensemble.

Now, in known dilution protocols, the amount of entanglement does not fluctuate, while the original entanglement concentration~\cite{BBPS1996} protocol, is one where the final amount of entanglement does fluctuate, but probabilistically rather than coherently. Nonetheless, due to the results in this section, Result \ref{result:bat-stochastic} still holds. In the protocol, of concentration, one starts with $n$ copies of
\begin{align}
\ket{\psi}_{AB} = \sqrt{p}\ket{00}_{AB}+\sqrt{1-p}\ket{11}_{AB}
\end{align}
and we want to concentrate them into $t$ copies
\begin{align}
\ket{\phi}_{AB} = \frac{1}{\sqrt{2}}\left(\ket{00}_{AB}
+\ket{11}_{AB}\right).
\end{align}
This can be done by having Alice perform a measurement onto projectors $P_t=\sum_{v\in t}\proj{v}$ where $v$ are strings in the Schmidt basis and $t$ is the set of all strings which have $t$ 1's (called the typecast). Since all strings which have the same typecast have equal probability, this projects the state into a maximally entangled one. The amount of entanglement gained $t$ is $nS(\psi_A)$ on average but has to satisfy the fluctuation theorems presented here.  And in fact, as we have shown, all other concentration schemes must also.

For entanglement dilution, existing protocols are not optimal, but using an entanglement battery, not only can one perform dilution on the single copy level, but also, the total yield can be improved and made optimal. Take for example, the teleportation protocol of \cite{BBPS1996}, where Alice performs Schumacher compression \cite{schumacher1996quantum} on her half of the $n$ copies of her local state so that it sits on only $k=nS(\psi_A)+O(\sqrt{n})$ qubits. She then teleports her state to Bob, using $k$ ebits. While the average number of ebits consumed is equal to $S(\psi_A)$, the total number requires an additional amount of order $\sqrt{n}$. This is due to the compression step, which although asymptotically efficient in terms of an average rate ($k/n$), wastes order $\sqrt{n}$ ebits. Likewise, the more sophisticated protocol of \cite{Lo-Popescu1999} also relies on compression and teleportation, using up $O(\sqrt{n})$ more ebits than strictly needed.  Dilution and concentration are thus not strictly reversible, since in concentration, the average is peaked around $S(\psi_A)$ but the amount fluctuates and can be both more or less than the average. However, using the entanglement battery, entanglement dilution can be performed in such a way that it is reversible, not only in the sense that the average amount of ebits consumed/produced do not differ by $O(1/\sqrt{n})$, but also in the sense that all the moments and the fluctuations are the same (as can be seen by applying Result \ref{res:rev}) which implies full reversibility.

\section{An estimate of the battery size}
\label{sec:finitebattery}

Here we give an example of a transition for which only a small battery is needed. For that we need the following lemma.

\begin{lemma}
	In a BLOCC protocol with a battery of finite size $N$, we have that
	\begin{equation}
	F(\tr_{A'B'}[\proj{\Phi_N}],\tr_{A'B'}[\tilde \sigma]) \ge \frac{1}{1+
		\frac{2  a^u_{\text{max}}}{N+1}},
	\end{equation}
	where $F(\rho,\sigma)$ is the quantum fidelity. The state $\tilde \sigma$ is defined as the output of the protocol when an initial product state between system and battery $\ket{\tilde \Psi_N}$ is the input (as defined in Eq. \eqref{eq:tpsin}) and $a^u_{\text{max}}$ is defined in Eq. \eqref{eq:maxau}. We recall that $\ket{\Phi_N}$ is a final state of the finite-sized LOCC transition as in Eq.   \eqref{eq:finalN}.
	 The marginals are over the system Hilbert space $AB$.
\end{lemma}
\begin{proof}
Starting from Eq. \eqref{eq:finitebound}, we have that the initial states $\ket{\Psi_N}$ and $\ket{\tilde \Psi_N}$ are close by
\begin{equation}
|\langle \tilde \Psi_N | \Psi_N \rangle |\ge \frac{1}{1+
	\frac{2 a^u_{\text{max}}}{N+1}} 
\end{equation}
Given the monotonicity property of the fidelity under CPTP maps, this quantity lower bounds that of the final states
$\Lambda( \proj{\Psi_N})=\proj{\Phi_N}$ and $\Lambda (\tilde \Psi_N)=\tilde \sigma$
\begin{equation}
F(\proj{\Phi_N},\tilde \sigma)=\sqrt{\bra{\Phi_N} \tilde \sigma \ket{\Phi_N}} \ge |\langle \tilde \Psi_N | \Psi_N \rangle |
\end{equation}
The fidelity between states can also only increase if we trace out the battery and focus on the system only
\begin{equation}
F(\tr_{A'B'}[\proj{\Phi_N}],\tr_{A'B'}[\tilde \sigma]) \ge F(\proj{\Phi_N},\tilde \sigma),
\end{equation}
 thus we get to the expression
\begin{equation}
F(\tr_{A'B'}[\proj{\Phi_N}],\tr_{A'B'}[\tilde \sigma]) \ge \frac{1}{1+
	\frac{2  a^u_{\text{max}}}{N+1}}.
\end{equation}
\end{proof}

The inequality in this lemma allows us to estimate how far we are from the ideal infinite battery case, in which all the different initial and final states coincide.

That is, to have a high fidelity, of at least $1-\epsilon$, the tradeoff between the size of the battery and the desired accuracy is
\begin{equation}
\frac{1}{1+
	\frac{2  a^u_{\text{max}}}{N+1}} \ge 1-\epsilon .
\end{equation}
We see that the parameter $a^u_{\text{max}}$ fixes the trade-off between size and fidelity. It is defined in Eq. \eqref{eq:maxau}, and is used to be able to approximate arbitrary values of $w$, so there is not a general upper bound for it. However, good particular choices of $w$, such as multiples of $\log{\frac{u}{u-1}}$ with $u$ integer, yield low $a^u_{\text{max}}$, and hence a good trade-off too. 

As a simple example, let us take a reversible transformation, in which $w\equiv w_{ij}=\log{\frac{p_i}{q_j}}$. In such processes $w_{\text{max}} = \max{\{D_\infty(p |q ), D_\infty (q|p) \}}$, where we define the Renyi-$\infty$ divergence as $D_\infty (p |q):=\log \sup_i (\frac{p_i}{q_i})$.
We choose an initial state with coefficients
$p=(1/2,1/4,1/8,1/8)$ and a final state with $q=(1/4,1/4,1/4,1/4)$, for which we have $w_{\text{max}} =\log 2$, and a choice of $u=2$ yields $ a^u_{\text{max}}=1$. 

Thus the states are
\begin{align}
\ket{\Psi_N}&= \sqrt{\frac{1}{2}}\ket{00} \otimes \sum_{x=\frac{n-N}{2}}^{\frac{n+N}{2}}\frac{1}{N+1} \ket{e_{x-1}}+\sqrt{\frac{1}{4}}\ket{11} \otimes \sum_{x=\frac{n-N}{2}}^{\frac{n+N}{2}}\frac{1}{N+1} \ket{e_{x}}
\\ &\,\,\,\,\,\,\,\,\,\,\,\,\,\,\,\,\,\,\,\,\,\,\,\,\,\,\,\,\,\,\,\,\,\,\,\,\,\,\,\,\,\,\,\,\,\,\,\,\,\,\,\,+\sqrt{\frac{1}{8}}(\ket{22}+\ket{33}) \otimes \sum_{x=\frac{n-N}{2}}^{\frac{n+N}{2}} \frac{1}{N+1} \ket{e_{x+1}}
\\
\ket{\Phi_N}&=\sqrt{\frac{1}{4}}(\ket{00}+\ket{11}+\ket{22}+\ket{33})\otimes \sum_{x'=\frac{n-N}{2}}^{\frac{n+N}{2}}\sqrt{\frac{1}{N+1}} \ket{e_{x'}}
\end{align}

For this case, to achieve a fidelity of at least $1-\epsilon = 0.85$, it is sufficient to take $N \ge 11$, and hence a battery consisting of $n \equiv N+2  a^u_{\text{max}}\ge 13$ systems. We also note that the dimension of the Hilbert space of the individual systems of the battery is $2u-1$, so in this case the Hilbert space of each of these is $\mathbb{C}^3\otimes \mathbb{C}^3$.

\end{document}